\documentclass[12pt]{article}

\usepackage{mathptmx}      
\usepackage[normalem]{ulem}

\usepackage{mathptmx}
\usepackage{upgreek}
\usepackage{mathtools}
\usepackage{times}
\usepackage{soul}
\usepackage{url}
\usepackage[utf8]{inputenc}
\usepackage{amsmath}
\usepackage{booktabs}
\usepackage{multirow}
\usepackage{tabularx}

\usepackage{amsthm}
\usepackage{amssymb}
\usepackage{hyperref}
\usepackage{colortbl}
\usepackage{xcolor}
\usepackage{graphicx}
\usepackage{enumitem}
\usepackage{nicefrac}
\renewenvironment{proof}[1][{\it{Proof}}]{\noindent\textbf{#1} }{\hfill \qed\medskip}
\newtheorem{observation}{Observation}
\newcommand{\edge}[2]{\{#1,#2\}}

\newcommand{\yes}{\mbox{Yes}}
\newcommand{\no}{\mbox{No}}
\newcommand{\yesins}{{\yes}-instance}
\newcommand{\noins}{{\no}-instance}
\newcommand{\poly}{\textsf{P}}
\newcommand{\polyy}{\textsf{Poly}}

\newcommand{\np}{\textsf{NP}}
\newcommand{\conp}{\textsf{CoNP}}
\newcommand{\nph}{{\np}-{hard}}
\newcommand{\nphshort}{{\np}-{h}}
\newcommand{\nphns}{\np-hardness}

\newcommand{\wbh}{\textsf{W[2]}-hard}
\newcommand{\wah}{\textsf{W[1]}-hard}

\newcommand{\bwahshort}{\textbf{\textsf{W[1]}-h}}

\newcommand{\wahns}{\textsf{W[1]}-hardness}
\newcommand{\wahshort}{{\textsf{W[1]}}-h}
\newcommand{\wbhshort}{{\textsf{W[2]}}-h}

\newcommand{\fpt}{{\textsf{FPT}}}

\newcommand{\prob}[1]{{\textsc{#1}}}
\newcommand{\memph}[1]{\emph{#1}}
\newcommand{\ww}{w}

\newcommand{\minimaxscore}[2]{\textsf{MAV}(#1,#2)} 
\newcommand{\ccavscore}[2]{\textsf{CCAV}(#1,#2)} 
\newcommand{\pavscore}[2]{\textsf{PAV}(#1,#2)} 

\newcommand{\hamdis}[2]{\mathcal{H}(#1,#2)}
\newcommand{\bigo}[1]{O(#1)}
\newcommand{\bigos}[1]{O^*(#1)}
\newcommand{\abs}[1]{|#1|}
\newcommand{\setmid}{:}
\newcommand{\R}{d}
\newcommand{\wind}[1]{{\sc{#1-Multiwinner}}}
\newcommand{\awind}[1]{\sc{Annotated #1-Multiwinner}}

\newcommand{\EP}[3]
{
\begin{center}
{
\begin{tabular}{|ll|} \hline
\multicolumn{2}{|l|}{\parbox[t]{0.96\columnwidth}{\textsc{#1}}}  \\ \hline
{\bf Input:}    & \parbox[t]{0.85\columnwidth}{#2\vspace*{1mm}}  \\
{\bf Question:} & \parbox[t]{0.85\columnwidth}{#3\vspace*{1mm}} \\ \hline
\end{tabular}
}
\end{center}
}

\newcommand{\EPP}[4]
{
\begin{center}
{
\begin{tabular}{|ll|} \hline
\multicolumn{2}{|l|}{\textsc{#1}}  \\ \hline 
{\bf Input:}        & \parbox[t]{0.83\columnwidth}{#2\vspace*{1mm}}  \\
{\bf Parameter:}    & \parbox[t]{0.83\columnwidth}{#3\vspace*{1mm}}  \\
{\bf Question:}     & \parbox[t]{0.83\columnwidth}{#4\vspace*{1mm}} \\ \hline
\end{tabular}
}
\end{center}
}

 \usepackage{geometry}
\newtheorem{lemma}{Lemma}
\newtheorem{corollary}{Corollary}
\newtheorem{theorem}{Theorem}
\newtheorem{claim}{Claim}
\newtheorem{openquestion}{Open Question}

\begin{document}

\title{Parameterized Complexity of Multiwinner Determination:\\
More Effort Towards Fixed-Parameter Tractability\thanks{A $3$-page extended abstract of the paper appeared in the proceedings of the \emph{17th International Conference on Autonomous Agents and Multiagent Systems} (AAMAS~2018)~\cite{DBLP:conf/atal/YangW18}.}
}

\author{Yongjie Yang$^a$ and Jianxin Wang$^{b,c}$}

\date{
\small{$^a$Chair of Economic Theory, Saarland University, Saarb\"{u}cken 66123, Germany \\
yyongjiecs@gmial.com\\[3pt]
$^b$School of Computer Science and Engineering, Central South University, Changsha 410083, China\\[3pt]
$^c$Xiangjiang Laboratory, Changsha 410205,China\\
{jxwang@mail.csu.edu.cn}
}
}

\maketitle

\begin{abstract}
We study the parameterized complexity of winner determination problems for three prevalent $k$-committee selection rules, namely the minimax approval voting (MAV), the proportional approval voting (PAV), and the Chamberlin-Courant's approval voting (CCAV). It is known that these problems are computationally hard. Although they have been studied from the parameterized complexity point of view with respect to several natural parameters, many of them turned out to be {\wah} or {\wbh}. Aiming at obtaining plentiful fixed-parameter algorithms, we revisit these problems by considering more natural single parameters, combined parameters, and structural parameters.
\medskip

\noindent{\bf{Keywords:}} {multiwinner voting, fixed-parameter tractability, minimax approval voting, proportional approval voting, Chamberlin-Courant's approval voting, treewidth, {\wah}}
\end{abstract}

\section{Introduction}
Committee selection rules (a.k.a.\ multiwinner voting rules) have received a considerable amount of attention recently due to their broad applications in social choice,
multi-agent systems, recommendation systems, etc.~\cite{trendscomsoc2017,DBLP:conf/eumas/GawronF22,DBLP:conf/agents/GhoshMHS99,DBLP:journals/jet/LacknerS21}. How efficiently  winning candidates with respect to a committee selection rule can be calculated is one of the most significant criteria to evaluate the applicability of these rules.
Many rules, such as STV, Bloc,~$k$-Borda, approval voting, satisfaction approval voting, admit polynomial-time algorithms for computing winners~\cite{DBLP:conf/atal/AzizGGMMW15,DBLP:journals/scw/ElkindFSS17}.
However, there are also salient committee selection rules with respect to which winners are {\nph} to compute.
Among them are particularly the minimax approval voting (MAV)~\cite{Bramsminimaxapproval2007,LeGrand2004Tereport}, the proportional approval voting (PAV)~\cite{DBLP:conf/atal/AzizGGMMW15,Kilgour2010,Thiele1985}, and the
Chamberlin-Courant's approval voting (CCAV)~\cite{ChamberlinC1983APSR10.2307/1957270,DBLP:journals/scw/ProcacciaRZ08}. Nevertheless, as these rules possess their own merits in many other aspects~\cite{DBLP:journals/scw/AzizBCEFW17,DBLP:conf/ijcai/FaliszewskiSST16,Kilgour2010,Kilgour2012,DBLP:journals/jet/LacknerS21}, researchers investigated the parameterized complexity of
winner determination problems under these rules, with the hope of procuring as many fixed-parameter tractability results as possible.
While many realistic parameters have been considered in the literature so far, there are still many relevant but underappreciated parameters. This paper aims to take a further step towards breaking the complexity barrier against the applicability of the above-mentioned three approval-based rules by extensively expanding the set of meaningful parameters leading to {\fpt}-algorithms for the winner determination problems.
Notably, in addition to many traditional parameters, we also study several structural parameters of the incidence graphs of approval-based elections. The incidence graph of an approval-based election is a bipartite graph whose vertex set is the candidate set union the voter set, and there is an edge between a candidate and a voter if and only if this voter approves the candidate.
To date, less is known about whether some structural parameters of the incidence graphs such as the treewidth and the size of a maximum matching lead to some {\fpt}-algorithms. These parameters are intimately connected to several important single parameters. For example, they are lower bounds of both the number of candidates~$m$ and the number of voters~$n$, and hence any {\fpt}-algorithm with respect to these two structural parameters directly carry over to~$m$ and~$n$.

\bigskip

{\it Organization.}
In Section~\ref{sec-pre}, we give definitions and notations used in the paper.
Then, we elaborate on related works and summarize our main results in Section~\ref{sec-related-work}.
Our concrete results are encapsulated in Sections~\ref{sec-sinlge-parameter}--\ref{sec-structural-parameter}. Specifically, Section~\ref{sec-sinlge-parameter} studies single parameters, Section~\ref{sec-combined-parameter} explores combinations of the single parameters occurring in Section~\ref{sec-sinlge-parameter}, and Section~\ref{sec-structural-parameter} focuses on the aforementioned structural parameters. For an overview of our main results, we refer to Table~\ref{tab-results}. We complete the paper by recapping our contributions and laying out several promising avenues for future research in Section~\ref{sec-conclusion}.

\section{Preliminaries}
\label{sec-pre}
In this section, we give essential notions related to our study.

\subsection{Elections}
\label{sec-elections}
We study {\it{approval-based multiwinner voting}}. In this setting, an {\it{election}} is a tuple $E=(C, V)$ where~$C$ is a set of candidates and~$V$ is a multiset of votes. Each vote of~$V$ is cast by a voter and is defined as a
subset of~$C$. In this paper, we interchangeably use the terms vote and voter. We say that a vote~$v$ {\it{approves}} a candidate~$c$ if $c\in v$. For each candidate~$c$,~$V(c)$ denotes the set of votes in~$V$ approving~$c$. Let~$k$ be a nonnegative integer. A {\it{$k$-set}} is a set of cardinality~$k$. A {\it{committee}} (respectively, {\it{$k$-committee}}) is a subset (respectively, $k$-subset) of candidates.
A {\it{$k$-committee selection rule}} maps each election $(C, V)$ and every nonnegative integer~$k$ such that $k\leq \abs{C}$ to a collection of $k$-committees of~$C$, {\it{winning $k$-committees}} of~$(C, V)$ under this rule.

The {\it{Hamming distance}} between two sets~$v$ and~$v'$ is defined as
\[\hamdis{v}{v'}=\abs{v\setminus v'}+\abs{v'\setminus v}.\]

We study the following~$k$-committee selection rules.

\begin{description}

\item[{\bf{MAV}}] In general, MAV selects~$k$-committees as close as possible to every vote, where the closeness is measured by the
Hamming distance. Precisely, the MAV score of a committee~$\ww$ with respect to~$(C,V)$ is $\minimaxscore{V}{\ww}=\max_{v\in V}\hamdis{v}{\ww}$. MAV selects~$k$-committees with the minimum MAV score.

\item[{\bf{CCAV}}]
A vote~$v$ is satisfied with a committee~$\ww$ if and only if at least one of~$v$'s approved candidates
is contained in~$\ww$, i.e., $v\cap \ww\neq \emptyset$. Candidates in $v\cap \ww$ are regarded as representatives of~$v$ in~$\ww$.
The CCAV score of a committee~$\ww$ with respect to~$(C,V)$ is $\ccavscore{V}{\ww}=\abs{\{v\in V \setmid v\cap \ww\neq\emptyset\}}$.
CCAV selects~$k$-committees with the maximum CCAV score.\footnote{CCAV is a special rule of the class of Chamberlin-Courant's rules.
See~\cite{DBLP:journals/scw/ElkindFSS17} for further discussions.}

\item[{\bf{PAV}}] The PAV score of a committee~$\ww$ with respect to an election~$(C, V)$ is $\pavscore{V}{\ww}=\sum_{v\in V, v\cap \ww\neq\emptyset} \sum_{i=1}^{\abs{v\cap \ww}}\frac{1}{i}$.
PAV selects~$k$-committees with the maximum PAV score.
\end{description}

Note that by the above definitions, $\ccavscore{V}{\emptyset}=\pavscore{V}{\emptyset}=0$ holds.

For each $\tau\in \{\text{MAV, CCAV, PAV}\}$, we study the following problem.

\EP
{\wind{$\tau$}}
{An election $E=(C, V)$, an integer $k\leq \abs{C}$, and a rational number~$\R$.}
{Is there a $k$-committee $\ww\subseteq C$ such that $\tau(V, \ww)\leq \R$ for
$\tau=\text{MAV}$, and $\tau(V,\ww)\geq \R$ for $\tau\in \{\text{CCAV, PAV}\}$?}

Throughout this paper, we study the following single parameters and consistently use the corresponding notations given below.

\begin{itemize}
\item $m=\abs{C}$.

\item $n=\abs{V}$.

\item $\R$: the threshold score of a desired  committee.

\item $k$: size of winning committees.

\item $\overline{k}=m-k$.

\item $\triangle_{\text{V}}=\max_{v\in V}\abs{v}$ is the maximum number of candidates a vote approves.

\item $\triangle_{\text{C}}=\max_{c\in C}\abs{V(c)}$ is the maximum number of votes approving a candidate in common.
\end{itemize}

The first four parameters are natural and have been substantially covered in the literature. (See Section~\ref{sec-related-work} for the details.)
The study of~$\overline{k}$ is motivated by the observation that in many real-life decision-making scenarios winners are picked by eliminating a small number of losers, or more relevantly, many decision-making processes are directly designed to select losers other than selecting winners. The last two parameters~$\triangle_{\text{V}}$ and~$\triangle_{\text{C}}$ have been either explicitly or implicitly studied in the literature. In many real-world applications voters are allowed to approve only a few candidates. In some other scenarios, voters are cognitively limited or time constrained so that they are only able to evaluate a small number of candidates. In these cases,~$\triangle_{\text{V}}$ is relatively small.
Regarding the parameter~$\triangle_{\text{C}}$, observe that it is no greater than~$n$. So, it is small whenever~$n$ is small. Moreover, as our goal is to provide a landscape of the parameterized complexity of {\wind{$\tau$}} as complete as possible, it is theoretically significant to study this parameter so as that we can offer results for combined parameters including~$\triangle_{\text{C}}$.

\subsection{Graphs}
We assume the reader is familiar with the basic concepts of graph theory, and refer  to~\cite{BrandstadtS,Douglas2000} for notions in graph theory used but not defined in the paper.
We only reiterate some basic notions below. A {\it{graph}} is a pair $(N, A)$ where~$N$ is a set of vertices and~$A$ is a set of edges over~$N$. Each {\it{edge}} is defined as an unordered pair of vertices, and between every two vertices there can be at most one edge. A {\it{multigraph}} is a generalization of a graph where between every two vertices there may exist multiple edges, and there may exist loops on vertices. A {\it{hypergraph}} is a generalization of a graph where every edge is a subset of vertices. A {\it{multihypergraph}} is a generalization of a hypergraph so that there may exist multiple edges consisting of the same vertices.

A {\it{matching}} of a graph $G=(N, A)$ is a subset~$M$ of~$A$ such that no two edges in~$M$ share a common vertex.
A vertex~$v$ is {\it{saturated}} by~$M$ if~$v$ is a vertex in some edge in~$M$.
A {\it{maximum matching}} of~$G$ is a matching with the maximum cardinality among all matching of~$G$. We use~$\alpha(G)$ to denote the size of a maximum matching of~$G$.

\subsection{Parameterized Complexity}

A parameterized problem is a subset of $\Upsigma^*\times \mathbb{N}$, where~$\Upsigma$ is a finite alphabet. A parameterized problem can be either {\it{fixed-parameter tractable}} ({\fpt}) or {\it{fixed-parameter intractable}}.
In particular, a parameterized problem is {\fpt} if there is an algorithm which correctly determines for each instance $(I, \kappa)$ of the problem whether $(I, \kappa)$ is a {\yesins} in time $\bigo{f(\kappa)\cdot \abs{I}^{\bigo{1}}}$, where~$f$ is a computable function and~$\abs{I}$ is the size of~$I$.
Fixed-parameter intractable problems are further classified into many classes including {\wah}, {\wbh}, etc.
For greater details on parameterized complexity theory, we refer to~\cite{DBLP:books/sp/CyganFKLMPPS15,DBLP:series/txcs/DowneyF13,DBLP:conf/dagstuhl/DowneyF92}.

\section{Related Works and Our Contributions}
\label{sec-related-work}
In this section, we discuss some important related works and outline our main contributions.

\subsection{Single Parameters}
{\wind{$\tau$}} has many natural parameters inherent in its definition, say, the seven single parameters listed at the end of Section~\ref{sec-elections}. All these seven parameters except~$\overline{k}$ have been explicitly or implicitly investigated in the literature prior to our work.

First, it is easy to see that {\wind{MAV}}, {\wind{CCAV}}, and {\wind{PAV}} are {\fpt} with respect to the parameter~$m$.
Misra, Nabeel, and Singh~\cite{DBLP:conf/atal/MisraNS15} proved that {\wind{MAV}} is {\fpt} with respect to the parameters~$\R$ and~$n$,
but becomes {\wbh} when parameterized by~$k$.
Betzler, Slinko, and Uhlmann~\cite{DBLP:journals/jair/BetzlerSU13} proved that {\wind{CCAV}} is {\fpt} with respect to the parameter~$n$,
but turned out to be {\wbh} when~$k$ is the parameter. Moreover, they considered a dual parameter $R=n-\R$.
They proved that {\wind{CCAV}} is {\nph} even when $R=0$, but presented an {\fpt}-algorithm with respect to the combined parameter
$k+R$.\footnote{A parameterized problem is {\fpt} with respect to the combination of two parameters~$\kappa$ and~$\kappa'$ if it is solvable in time $\bigos{f(\kappa,\kappa')}$ where~$f$ is a computable function in~$\kappa$ and~$\kappa'$, or equivalently it is {\fpt} with respect to $\kappa+\kappa'$.}
Aziz~et~al.~\cite{DBLP:conf/atal/AzizGGMMW15} proved that {\wind{PAV}} is {\wah} with respect to~$k$ even if every voter approves two
candidates.

We complement these results as follows. First, we close some gaps and improve an {\fpt}-algorithm.
Concretely, we propose an {\fpt}-algorithm for {\wind{PAV}} when parameterized by~$n$. We also observe that {\wind{CCAV}} is equivalent to the {\prob{Partial Hitting Set}} problem, and as a consequence of a previous result for the {\prob{Partial Hitting Set}} problem, {\wind{CCAV}} can be solved in~$\bigos{2^{\bigo{\R}}}$ time or in $\bigos{2^{\bigo{n}}}$ time.\footnote{$\bigos{}$ is~$\bigo{}$ with polynomial factors being omitted.}
It should be noted that our new observation-based algorithm substantially improves the previous best {\fpt}-algorithm for {\wind{CCAV}} parameterized by~$n$ studied in~\cite{DBLP:journals/jair/BetzlerSU13} which runs in time~$\bigos{n^n}$.
Second, we study the parameter $\overline{k}=m-k$, the number of nonwinning candidates.
We show that {\wind{MAV}}, {\wind{CCAV}}, and {\wind{PAV}} are all {\wah} with respect to this parameter, even when every voter approves two candidates.
Third, from previous results by other researchers, we achieve numerous dichotomy results with respect to the two natural parameters~$\triangle_{\text{V}}$ and~$\triangle_{\text{C}}$.
It is known that {\wind{MAV}},  {\wind{CCAV}}, and {\wind{PAV}} are already {\nph} when $\triangle_{\text{V}}=2$ and
$\triangle_{\text{C}}=3$~\cite{DBLP:conf/atal/AzizGGMMW15,LeGrand2004Tereport,DBLP:conf/ijcai/ProcacciaRZ07}.
We prove that {\wind{MAV}} and {\wind{CCAV}} become polynomial-time solvable if $\triangle_{\text{V}}\leq 1$ or $\triangle_{\text{C}}\leq 2$,
and {\wind{PAV}} becomes polynomial-time solvable if $\min\{\triangle_{\text{V}}, \triangle_{\text{C}}\}=1$ or $\triangle_{\text{V}}=\triangle_{\text{C}}=2$.

\subsection{Combined Parameters}
Obviously, if a problem is {\fpt} with respect to a parameter~$\kappa$, it is {\fpt} with respect to any combined parameter including~$\kappa$.
As except {\wind{PAV}} with respect to~$\R$ whose fixed-parameter tractability is open, {\wind{MAV}}, {\wind{CCAV}}, and {\wind{PAV}} are {\fpt} with respect to~$m$,~$n$, and~$\R$,
it only makes sense to study combinations of other parameters. As $k+\overline{k}=m$,
{\wind{$\tau$}} for $\tau\in \{\text{MAV}, \text{CCAV}, \text{PAV}\}$ is {\fpt} with respect to $k+\overline{k}$.
For MAV and CCAV, the remaining combinations of two single parameters are $k+\triangle_{\text{V}}$, $k+\triangle_{\text{C}}$, $\overline{k}+\triangle_{\text{V}}$, and $\overline{k}+\triangle_{\text{C}}$.
We establish many {\fpt}-results with respect to these combined parameters.
Concretely, we obtain {\fpt}-results for {\wind{MAV}} and {\wind{CCAV}} when parameterized by $k+\triangle_{\text{C}}$ and $\overline{k}+\triangle_{\text{C}}$.
However, as {\wind{MAV}}, {\wind{CCAV}}, and {\wind{PAV}} are {\wah} with respect to~$\overline{k}$ (even when every vote approves two candidates), they are {\wah} when parameterized by $\overline{k}+\triangle_{\text{V}}$. For the parameter $k+\triangle_{\text{V}}$, we develop an {\fpt}-algorithm for {\wind{MAV}}, but we show that {\wind{CCAV}} is {\wah} even when every vote approves two candidates.
Concerning PAV, a reduction by Aziz~et~al.~\cite{DBLP:conf/atal/AzizGGMMW15} implies that {\wind{PAV}} is {\wah} with respect to $k+\triangle_{\text{V}}$.
 We are unable to prove the {\fpt}-membership of {\wind{PAV}} with respect to~$\R$, but we show that combining~$\R$ and~$\triangle_{\text{V}}$ leads to an {\fpt}-result.

We would like to point out that Misra, Nabeel, and Singh~\cite{DBLP:conf/atal/MisraNS15} studied kernelization of {\wind{MAV}} with respect to
the combined parameters $\R+m$ and $n+k$, and showed that these problems do not admit any polynomial kernels unless $\conp\subseteq \np/\polyy$\footnote{A kernelization with respect to a parameter is a polynomial-time algorithm which transforms an instance into an equivalent instance with the size of the main part being bounded from above by a computable function of the parameter. See~\cite[Chapter~2]{DBLP:books/sp/CyganFKLMPPS15} for more details.}.
In addition, Liu and Guo~\cite{DBLP:conf/atal/LiuG16} studied the combined parameter $k+n$ for some generalizations of {\wind{MAV}}.
Betzler, Slinko, and Uhlmann~\cite{DBLP:journals/jair/BetzlerSU13}  proved that {\wind{CCAV}} is {\wbh} with respect to the combined parameter of $R=n-\R$ and~$k$. Though not explicitly stated, their reduction actually implies that {\wind{CCAV}} does not admit any polynomial kernel with respect to both~$n$ and~$m$, unless the polynomial hierarchy collapses to the third level.\footnote{These kernelization lower bounds follow from a reduction by  Betzler, Slinko, and Uhlmann~\cite{DBLP:journals/jair/BetzlerSU13}  and several techniques for establishing kernelization lower bounds delineated in~\cite[Chapter~14]{DBLP:books/sp/CyganFKLMPPS15} and~\cite{DBLP:journals/talg/DomLS14}.}

\subsection{Structural Parameters}
Heretofore, the most widely studied structural parameters in the setting of multiwinner voting are based on various concepts of restricted preference domains such as single-peaked preferences and single-crossing preferences (see, e.g.,~\cite{DBLP:conf/ecai/CornazGS12}).
Particularly, these parameters measure how far an election is away from a specific domain of preferences.
In this paper, we study two structural parameters of incidence graphs of elections.
Recall that the {\it{incidence graph}} of an election $E=(C, V)$ is a bipartite graph~$G_E$ with vertex set $C\cup V$ so that there is an edge between a candidate~$c\in C$ and a vote~$v\in V$ if and only if~$c\in v$.
We prove that {\wind{CCAV}} is {\fpt} with respect to treewidth of incidence graphs, and {\wind{MAV}} and {\wind{PAV}} are {\fpt} if we combine the treewidth and the parameter~$k$.
When parameterized by the size of maximum matchings of incidence graphs, we have {\fpt}-algorithms for all three rules.

\subsection{Other Related Works}
In addition to the extensive effort made from the angle of parameterized complexity, much exploration on the complexity of {\wind{$\tau$}} restricted to
preference domains has been pursued over the past few years.
Betzler, Slinko, and Uhlmann~\cite{DBLP:journals/jair/BetzlerSU13}  developed polynomial-time algorithms for {\wind{CCAV}} in the single-peaked domain.
This algorithm was subsequently extended to an {\fpt}-algorithm with respect to the parameter single-peaked width by Cornaz, Galand, and Spanjaard~\cite{DBLP:conf/ecai/CornazGS12}.
Yu, Chan, and Elkind~\cite{DBLP:conf/ijcai/YuCE13} studied the domain of single-peaked on trees, and obtained both polynomial-time algorithms and {\nphns} results for many variants of CCAV. One of their results~\cite[Theorem~4.1]{DBLP:conf/ijcai/YuCE13} is in essence an {\fpt}-algorithm for {\wind{CCAV}} with respect to the combined parameter $m+k+\lambda$, where~$\lambda$ is the number of leaves of the underlying tree. A follow-up paper by Peters and Elkind~\cite{DBLP:conf/aaai/PetersE16} addressed several open questions left in~\cite{DBLP:conf/ijcai/YuCE13}.
Later, Elkind and Lackner~\cite{DBLP:conf/ijcai/ElkindL15} proposed~13 different restricted domains of dichotomous preferences, and obtained several polynomial-time algorithms and {\fpt}-algorithms for {\wind{$\tau$}} for~$\tau$ being CCAV, MAV, and PAV with respect to the parameter~$\triangle_{\text{V}}$, and the combined parameter~$k+\R$, when the given elections fall into certain specific categories of their proposed domains.
Liu and Guo~\cite{DBLP:conf/atal/LiuG16} presented polynomial-time algorithms for {\wind{MAV}} in two of the domains
proposed by Elkind and Lackner. Peters and Lackner~\cite{DBLP:journals/jair/PetersL20} proposed the domain of single-peaked on a circle, an appealing generalization of the single-peaked domain.
They provided a polynomial-time algorithm for {\wind{CCAV}} in this domain.
Peters~\cite{DBLP:conf/aaai/Peters18} observed a relationship between voting problems restricted to the single-peaked domain and totally unimodular integer linear programming. In view of this observation and the fact that totally unimodular integer linear programming is
polynomial-time solvable~\cite{Schrijver:1986:TLI:17634}, a variety of polynomial-time solvability results were obtained, including one for {\wind{PAV}}\@.\footnote{We note that~\cite{DBLP:journals/jair/PetersL20} integrates a conference paper with the same title appeared in AAAI 2017 and~\cite{DBLP:conf/aaai/Peters18}.}
Skowron~et~al.~\cite{DBLP:journals/tcs/SkowronYFE15} complemented these results by showing that {\wind{CCAV}} is polynomial-time solvable restricted
to the single-crossing domain. Clearwater, Puppe, and Slinko~\cite{DBLP:conf/ijcai/ClearwaterPS15} then extended this result to the domain of single-crossing on trees.
Yang~\cite{DBLP:conf/ijcai/Yang19a} expanded several domains of Elkind and Lackner~\cite{DBLP:conf/ijcai/ElkindL15} to directed tree-embedded domains, and derived a number of polynomial-time algorithms for {\wind{MAV}}, {\wind{CCAV}}, and {\wind{PAV}} restricted to these domains.

As traditional approaches tackling {\nph} problems, approximation and heuristic algorithms for {\wind{$\tau$}} where $\tau \in \{\text{MAV}, \text{CCAV}, \text{PAV}\}$ have also  been perpetually reported over the past few years~\cite{DBLP:conf/wine/ByrkaS14,DBLP:journals/jair/CyganKSS18,DBLP:journals/heuristics/FaliszewskiSST18,DBLP:journals/iandc/Skowron17,DBLP:journals/jair/SkowronF17}.

We would like to mention that the (parameterized) complexity of {\wind{$\tau$}} for numerous ranking-based multiwinner voting rules~$\tau$ has been considerably studied in the literature, too~\cite{DBLP:journals/jair/BetzlerSU13,DBLP:journals/scw/FaliszewskiSST18,DBLP:conf/ijcai/Gupta00T21}. On top of that, the (parameterized) complexity of winner determination problems for several variants of multiwinner voting rules has been explored recently~\cite{DBLP:conf/aaai/BredereckF0KN20,DBLP:conf/aaai/000121a}.

Finally, we remark that, besides winner determination problems, investigating the complexity of many strategic voting problems, such as manipulation, control, and bribery, has gained increasing interest in recent years as well~\cite{DBLP:journals/ai/BredereckFKNST21,DBLP:conf/atal/FaliszewskiST17,DBLP:conf/atal/000120,DBLP:conf/ijcai/Yang19}. We refer to~\cite{DBLP:series/sbis/LacknerS23} for a comprehensive survey of approval-based multiwinner voting where many other important issues not discussed above have been greatly elaborated on.

\begin{table}[ht!]
\renewcommand\arraystretch{1.2}
\renewcommand{\tabcolsep}{1.3mm}
\caption{A summary of the parameterized complexity of {\wind{$\tau$}}. Our results are in boldface.
Additionally,~$\omega$ is the treewidth and~$\alpha$ is the size of a maximum matching of the incidence graph of a given election.
Hardness results marked by the~$\star$ symbol mean that they hold even if every vote approves two candidates. Results without any marks mean that they are either trivial or implied by other results in the table. Underlined {\fpt}-results mean that the corresponding problems do not admit any polynomial kernels assuming certain standard complexity hypothesis.}
{\centering{
\begin{tabular}{|l|c|c|c|c|c|c|}\hline
&\multicolumn{6}{c|}{single parameters}\\ \cline{2-7}
&&&&&&\\ [-8pt]

&$\R$
&$m$
&$n$
&$k$
&$\overline{k}$
&{$(\triangle_{\text{V}}, \triangle_{\text{C}})$}
\\ \hline

\multirow{2}{*}{MAV}
&\multirow{2}{*}{{\underline{\fpt}} \protect\cite{DBLP:conf/atal/MisraNS15}}
&\multirow{2}{*}{{\underline{\fpt}} \protect\cite{DBLP:conf/atal/MisraNS15}}
&\multirow{2}{*}{\underline{\fpt} \protect\cite{DBLP:conf/atal/MisraNS15}}
&\multirow{2}{*}{{\wbhshort} \protect\cite{DBLP:conf/atal/MisraNS15}}
&{\bf{\wahshort}}$^{\star}$
&{$(\geq 2, \geq 3)$:} {\nphshort~\cite{LeGrand2004Tereport}}\\

&
&
&
&
&(Cor.~\ref{cor-mav-wah-bark})
& {others: \bf\poly} (Cor.~\ref{cor-poly-every-voter-approves-one-candidate}, Thm.~\ref{thm-maximin-poly-cases})\\ \hline

\multirow{2}{*}{CCAV}
&$\mathbf{{2^{\bigo{\R}}}}$
&\multirow{2}{*}{\underline{\fpt}  \protect\cite{DBLP:journals/jair/BetzlerSU13}}
&$\underline{n^n}$ \protect\cite{DBLP:journals/jair/BetzlerSU13}
&{\wbhshort} \protect\cite{DBLP:journals/jair/BetzlerSU13}
&{\bf\wahshort}$^{\star}$
&$(\geq 2, \geq 3)$: {\nphshort~\protect\cite{DBLP:conf/ijcai/ProcacciaRZ07}}\\

&(Cor.~\ref{thm-ccav-R-FPT})
&
&$\mathbf{{2^{\bigo{n}}}}$ (Cor.~\ref{cor-ccav-fpt-n})
& {\bwahshort}$^{\star}$ (Thm.~\ref{thm-ccav-wah-k-voter-approve-two-candidates})
& (Thm.~\ref{thm-ccav-wah-dual-k})
&{others: \bf\poly} (Cor.~\ref{cor-poly-every-voter-approves-one-candidate}, Thm.~\ref{thm-maximin-poly-cases})\\ \hline

\multirow{3}{*}{PAV}
&\multirow{3}{*}{open}
&\multirow{3}{*}{\fpt}
&{\bf\fpt}
&
&{\bf\wahshort}$^{\star}$
&{$(\geq 2, \geq 3)$:} {\nphshort~\protect\cite{DBLP:conf/atal/AzizGGMMW15}}\\

&
&
&(Thm.~\ref{thm-pav-n-fpt})
&{\wahshort}$^{\star}$ \protect\cite{DBLP:conf/atal/AzizGGMMW15}
& (Thm.~\ref{thm-pav-wah-dual-k})
&{$(\geq 3, 2)$: open}\\

&
&
&
&
&
&{others: \bf\poly} (Cor.~\ref{cor-poly-every-voter-approves-one-candidate}, Thm.~\ref{thm-pav-p-special})\\ \hline
\end{tabular}
}}

\medskip
{\centering{
\renewcommand\arraystretch{1.2}
\renewcommand{\tabcolsep}{1.93mm}
\begin{tabular}{|l|c|c|c|c|c|c|c|}\hline
&\multicolumn{5}{c|}{combined parameters}
&\multicolumn{2}{c|}{structural parameters} \\ \cline{2-8}  &&&&&&&\\[-8pt]

&$k+\triangle_{\text{C}}$
&$k+\triangle_{\text{V}}$
&$\overline{k}+\triangle_{\text{C}}$
&$\overline{k}+\triangle_{\text{V}}$
&$\R+\triangle_{\text{V}}$
&$\omega$
&$\alpha$ \\ \hline

\multirow{2}{*}{MAV}
&{\bf{\fpt}}
&{\bf{\fpt}}
&{\bf\fpt}
&\multirow{2}{*}{\wahshort$^{\star}$}
&\multirow{2}{*}{\fpt}
& ${\omega}+k$
&{\bf\fpt} \\

&(Thm.~\ref{thm-mav-fpt-k-triangle-c})
&(Cor.~\ref{thm-mav-fpt-k-triangle-v})
&(Thm.~\ref{thm-mav-fpt-bar-k-triangle-c})
&
&
&{\bf{\fpt}} (Thm.~\ref{pav-mav-fpt-k-tree-width})
&(Thm.~\ref{thm-mav-pav-matching})\\
\hline

\multirow{2}{*}{CCAV}
&{\bf{\fpt}}
&\multirow{2}{*}{{\wahshort}$^{\star}$}
&{\bf{\fpt}}
&\multirow{2}{*}{{\wahshort}$^{\star}$}
&\multirow{2}{*}{{\fpt}}
&\multirow{2}{*}{$\mathbf 4^{\omega}$ (Thm.~\ref{thm-cav-fpt-treewidth})}
&\multirow{2}{*}{$\mathbf 4^{\alpha}$}\\

&(Cor.~\ref{thm-ccav-fpt-k-triagnel-c})
&
&(Thm.~\ref{thm-ccav-wd-fpt-bar-k-c})
&
&
&
&\\  \hline

\multirow{2}{*}{PAV}
&$k+\triangle_{\text{C}}+\triangle_{\text{V}}$
&\multirow{2}{*}{\wahshort$^{\star}$\ \protect\cite{DBLP:conf/atal/AzizGGMMW15}}
&\multirow{2}{*}{open}
&\multirow{2}{*}{{\wahshort}$^{\star}$}
&{\bf{\fpt}}
&$\omega+k$
&{\bf\fpt}\\

&{\bf{\fpt}} (Cor.~\ref{thm-pav-fpt-k-c-v})
&
&
&
&(Thm.~\ref{thm-pav-fpt-r-v})
&{\bf{\fpt}} (Thm.~\ref{pav-mav-fpt-k-tree-width})
&(Thm.~\ref{thm-mav-pav-matching})\\ \hline
\end{tabular}
}} 
\label{tab-results}
\end{table}

\section{Single Parameters}
\label{sec-sinlge-parameter}
In this section, we investigate some predominant single parameters.

\subsection{Parameters~\texorpdfstring{$\R$,~$n$}{Lg}, and~\texorpdfstring{$\overline{k}$}{Lg}}

We start with an {\fpt}-algorithm for {\wind{PAV}} with respect to~$n$.
Before presenting our algorithm, let us recall the main idea of an {\fpt}-algorithm for {\wind{MAV}} with respect to~$n$~\cite{DBLP:conf/atal/MisraNS15}.
First, the candidates are partitioned into at most~$2^n$ subsets, each consisting of all candidates approved by exactly the same votes.
Then, what matters for solving the problem is only how many candidates from each subset are contained in a desired $k$-committee.
Based on this observation, {\wind{MAV}} can be reduced to integer linear programming (ILP) by assigning to each subset defined above an integer variable.
As the number of variables is bounded by~$2^n$, Lenstra's theorem~\cite{Lenstra83} provides an {\fpt}-algorithm for {\wind{MAV}} with respect to~$n$.
Although the {\fpt}-algorithm with respect to~$n$ for {\wind{CCAV}} derived by Betzler, Slinko, and Uhlmann~\cite{DBLP:journals/jair/BetzlerSU13} is not ILP-based, it is easy to see that  {\wind{CCAV}} also admits a similar ILP formulation.
Unfortunately, this framework does not apply to {\wind{PAV}}\@. The reason is that the objective in this case is a nonlinear function.
To overcome this obstacle, we resort to an {\fpt}-framework proposed by Bredereck~et~al.~\cite{DBLP:journals/tcs/BredereckFNST20} for the {\sc{Mixed Integer Programming With Simple Piecewise Linear Transformations}} problem.
In fact, we need only a special case of the problem defined below.
For a vector $x\in \mathbb{Z}^p$ and an integer $i\in [p]$, we use~$x_i$ to denote the $i$-th component of~$x$.

A real-valued function~$f$ is {\it{concave}} if for every~$x$ and~$y$ such that $x<y$ and every~$\lambda$ such that $0\leq \lambda \leq 1$ it holds that
\[f(\lambda \cdot x + (1-\lambda) \cdot y)\geq \lambda \cdot f(x)+ (1-\lambda) \cdot f(y).\]
Intuitively, a function is concave if for every~$x$,~$y$, and~$z$ such that $x\leq y\leq z$, the point $(y, f(y))$ is not below the straight line determined by the two points $(x, f(x))$ and $(z, f(z))$.
A {\it{piecewise linear concave function}} is a piecewise linear function that is concave.

\EP
{Integer Programming With Simple Piecewise Linear Transformations (IPWSPLT)}
{A collection $\{f_{i,j} \setmid i\in [p], j\in [q]\}$ of $p\cdot q$ piecewise linear concave functions, and a vector $b\in \mathbb{Z}^p$.}
{Is there a vector $x\in \mathbb{Z}^q$ such that for every~$i\in [p]$ it holds that
\begin{equation}
\label{equ-generalized-ilp}
\sum_{j=1}^q f_{i,j}(x_j)\leq b_i?
\end{equation}}

The original problem MIPWSPLT studied by Bredereck~et~al.~\cite{DBLP:journals/tcs/BredereckFNST20} is more general in that it allows the existence of additional variables which may take nonintegral values, and allows the occurrences of both piecewise linear concave functions and piecewise linear convex functions simultaneously.

\begin{lemma}
[\cite{DBLP:journals/tcs/BredereckFNST20}]
\label{lem-generalized-ILP}
{\memph{IPWSPLT}} can be solved in time $\bigo{{\emph{\textsf{poly}}}(\abs{I}, t) \cdot q^{2.5q+o(q)}}$, where~$\abs{I}$ is the number of bits encoding the input, and~$t$ is the maximum number of pieces per function.
\end{lemma}

We note that Lemma~\ref{lem-generalized-ILP} still holds
if in the definition of IPWSPLT the less than sign is replaced with the greater than sign or the equal sign in~\eqref{equ-generalized-ilp} for several $i\in [p]$~\cite{DBLP:journals/tcs/BredereckFNST20}.

Now we are ready to present our {\fpt}-algorithm for {\wind{PAV}} with respect to~$n$.
We provide indeed an algorithm for a more general problem called {\prob{Annotated }\wind{PAV}}. In this problem, we are given an election $(C, V)$, a subset $C'\subseteq C$ of candidates, an integer~$k$ such that $\abs{C'}\leq k\leq \abs{C}$, and a number~$\R$, and the question is whether there is a $k$-committee $\ww\subseteq C$ such that $C'\subseteq \ww\subseteq C$ and {\textsf{PAV}}$(V, \ww)\geq \R$.  Clearly, {\wind{PAV}} is a special case of {\awind{PAV}} where $C'=\emptyset$.
This generalization is afterward exploited in an algorithm presented in Section~\ref{sec-structural-parameter}.

\begin{figure}[h!]
\begin{center}
\includegraphics[width=0.45\textwidth]{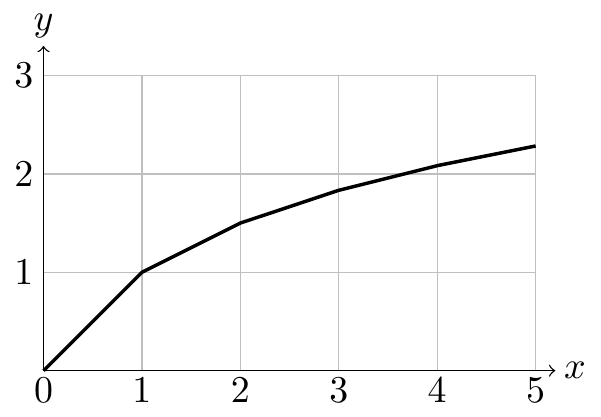}
\end{center}
\caption{An illustration of the piecewise linear concave function~$f$ in the proof of Theorem~\ref{thm-pav-n-fpt}.}
\label{fig-concave-function}
\end{figure}

\begin{theorem}
\label{thm-pav-n-fpt}
{\awind{PAV}} is {\memph\fpt} with respect to~$n$.
\end{theorem}

\begin{proof}
We prove Theorem~\ref{thm-pav-n-fpt} by reducing {\awind{PAV}} to {\prob{IPWSPLT}}.

Let $I=(E, C', k, \R)$ be an instance of {\awind{PAV}}, where $E=(C, V)$ and $C'\subseteq C$. Let $n=\abs{V}$ be the number of votes. In what follows, we  construct an integer programming formulation of {\awind{PAV}}, prove the correctness of the reduction, and prove that the formulation is an instance of {\prob{IPWSPLT}}.
We create two types of variables. First, we create a variable~$x_v$ for each vote~$v\in V$, indicating the number of~$v$'s approved candidates
in a desired~$k$-committee. Second, for each $U\subseteq V$, let~$C_U$ be
the set of candidates approved by all votes in~$U$ but disapproved by votes from~$V\setminus U$, i.e., $C_U=\{c\in C \setmid V(c)=U\}$.
We create a variable~$x_U$ for each~$U\subseteq V$ which indicates the number of candidates from~$C_U$ that are supposed to be in a desired~$k$-committee.
The constraints are as follows.
\renewcommand\labelenumi{(\theenumi)}
\begin{enumerate}
\item For each variable~$x_U$,~$U\subseteq V$, it holds that $\abs{C'\cap C_U}\leq x_U\leq \abs{C_U}$.

\item For each~$v\in V$, we have $x_v=\sum_{v\in U\subseteq V} x_U$.

\item As we aim to select a committee of cardinality~$k$, we have $\sum_{U\subseteq V}x_U=k$.

\item For the last constraint, we need to define a piecewise linear concave function $f: \mathbb{R}_{\geq 0}\rightarrow \mathbb{R}_{\geq 0}$ as follows.
First, $f(0)=0$. Second, for each positive integer~$x$, we define $f(x)=\sum_{i=1}^x\frac{1}{i}$.
Third, for each real~$x$ such that $y< x< y+1$ and~$y$ is a nonnegative integer, we define
\[f(x)=f(y)+(x-y)\cdot (f(y+1)-f(y)).\]
Figure~\ref{fig-concave-function} illustrates the function~$f$.

The last constraint is then $\sum_{v\in V} f(x_v)\geq \R$, which ensures that the desired $k$-committee has PAV score at least~$\R$.
\end{enumerate}

Now we show the correctness of the reduction, i.e., we show that the given instance~$I$ is a {\yesins} if and only if the above integer programming has a feasible solution.

$(\Rightarrow)$ Assume that~$I$ is a {\yesins}, i.e., there exists a $k$-committee~$\ww$ so that $C'\subseteq \ww\subseteq C$ and $\pavscore{V}{\ww}\geq \R$. We assign to the variables constructed above the corresponding values, i.e., we let~$x_{U}=\abs{\ww \cap C_U}$ for each $U\subseteq V$, and let $x_v=\abs{v\cap \ww}$ for each $v\in V$. Obviously, $\abs{\ww\cap C_U}\leq \abs{C_U}$. As $C'\subseteq \ww$, it holds that $\abs{C'\cap C_U}\leq \abs{\ww\cap C_U}$. Therefore, constraints defined in~(1) hold. As for every two distinct $U, U'\subseteq V$,~$C_U$ and~$C_{U'}$ are disjoint and, moreover, $\bigcup_{{U}\subseteq V}C_{{U}}=C$, it holds that
\begin{equation}
\label{eq-c}
\abs{v\cap \ww}=\sum_{v\in U\subseteq V} \abs{v\cap \ww\cap C_U}=\sum_{v\in U\subseteq V} \abs{\ww\cap C_U}=\sum_{v\in U\subseteq V} x_U
\end{equation}
for all $v\in V$, and $\sum_{U\subseteq V}\abs{\ww\cap C_U}=\abs{\ww}=k$. In other words, all constraints described in~(2)--(3) are satisfied. Notice that by the definition of~$C_U$, if $v\not\in U$ then $C_U\cap v=\emptyset$, and if $v\in U$ then $C_U\subseteq v$. This ensures the correctness of Equality~\eqref{eq-c}.  Finally, the constraint given in~(4) holds because $\sum_{v\in V} f(x_v)$ is exactly the PAV score of~$\ww$ with respect to~$V$ which is equal to or greater than~$\R$.

$(\Leftarrow)$ Assume that the above integer programming has a feasible solution. We show below that~$I$ is a {\yesins} by constructing a desired $k$-committee~$\ww$. Initially, let $\ww=\emptyset$. For every~$x_U$, $U\subseteq V$, in the feasible solution, we arbitrarily select~$x_U$ candidates from~$C_U$ so that all candidates in $C'\cap C_U$ are selected, and add them into~$\ww$. By the constraints in~(1), these~$x_U$ candidates exist for each $U\subseteq V$. By the constraint in~(3), we know that~$\ww$ consists of~$k$ candidates. By the constraints in~(2),~$x_v$ is exactly the number of candidates from~$\ww$ approved in~$v$. As a result, $\sum_{v\in V} f(x_v)$ is exactly the PAV score of~$\ww$ with respect to~$V$. Then, from the constraint in~(4), it follows that $\pavscore{V}{\ww}\geq \R$. Finally, as $\bigcup_{U\subseteq V}C_U=C$, by the definition of~$\ww$ we know that~$C'$ is contained in~$\ww$. Now we can conclude that~$I$ is a {\yesins} of {\awind{PAV}}.

Next, we show that the above integer programming is an instance of {\prob{IPWSPLT}}. To this end, we reiterate that {\prob{IPWSPLT}} contains ILP as a special case. Constraints described in (1)--(3) are standard constraints (or can be transformed into standard form trivially) of ILP. It is easy to see that the function~$f$ defined in~(4) is a piecewise linear concave function. Therefore, the above integer programming is an instance of {\prob{IPWSPLT}}.

It remains to analyze the running time of the algorithm. The above reduction clearly takes {\fpt}-time in~$n$, and the number of variables is $n+2^n$. Then, by Lemma~\ref{lem-generalized-ILP} and the correctness of the reduction, {\awind{PAV}} is {\fpt} with respect to~$n$.
\end{proof}

We note that a similar {\fpt}-algorithm for winners computation for a large class of ranking-based multiwinner voting rules has been also derived by Faliszewski~et~al.~\cite[Theorem~16]{DBLP:journals/scw/FaliszewskiSST18}.

Now we study the parameter~$\overline{k}=m-k$, i.e., the number of candidates not in a desired $k$-committee.
In regard to MAV, the {\nphns} proof by LeGrand~\cite{LeGrand2004Tereport} actually already implied that {\wind{MAV}} is {\wah} with respect to~$\overline{k}$. Precisely, LeGrand's {\nphns}  reduction is from the {\sc{Vertex Cover}} problem, where vertices correspond to candidates, and edges correspond to votes so that each vote approves exactly the candidates corresponding to its two endpoints. It is easy to see that there is a vertex cover of cardinality~$\kappa$ if and only if there is a $\kappa$-committee so that the Hamming distance between this committee and each vote is at most~$\kappa$. The {\wahns} of {\wind{MAV}} follows from that {\sc{Vertex Cover}} is {\wah} with respect to $m-\kappa$ where~$m$ is the number of vertices (candidates)~\cite{DBLP:journals/tcs/DowneyF95}.

\begin{corollary}[\cite{DBLP:journals/tcs/DowneyF95,LeGrand2004Tereport}]
\label{cor-mav-wah-bark}
{\wind{MAV}} is {\memph\wah} with respect to~$\overline{k}$ even when every vote approves two candidates.
\end{corollary}

It should be noted that {\sc{Vertex Cover}} with respect to $m-\kappa$ is exactly a parameterized variant of the {\sc{Independent Set}} problem. We provide the formal definition of the problem below because it will be used to establish our next intractability result.

An {\it{independent set}} of a graph is a subset of pairwise nonadjacent vertices.

\EPP
{$\kappa$-Independent Set ($\kappa$-IDS)}
{A graph $G=(N, A)$ and an integer $\kappa$.}
{$\kappa$.}
{Does $G$ admit an independent set of cardinality at least $\kappa$?}

\begin{theorem}
\label{thm-ccav-wah-dual-k}
{\wind{CCAV}} is {\memph\wah} with respect to~$\overline{k}$, even when every vote approves two candidates.
\end{theorem}

\begin{proof}
To prove Theorem~\ref{thm-ccav-wah-dual-k}, we offer a reduction from {\sc{$\kappa$-IDS}} to {\wind{CCAV}}, similar to the one for {\wind{MAV}} discussed above by LeGrand~\cite{LeGrand2004Tereport}.
Let $(G=(N, A), \kappa)$ be an instance of {\sc{$\kappa$-IDS}}.
We construct an instance $((C, V), k, \R)$ of {\wind{CCAV}} as follows.
For each vertex~$c\in N$, we create one candidate denoted by the same symbol for simplicity. Let $C=N$ and let $m=\abs{C}$.
For each edge $\edge{c}{c'}\in A$, we create one vote~$v$ approving exactly~$c$ and~$c'$, i.e., $v=\{c, c'\}$. Let~$V$ be the set of the created votes.
Finally, we set $\overline{k}=\kappa$ (hence $k=m-\kappa$) and set $\R=\abs{A}$.

The correctness of the reduction is easy to check. If there is an independent set~$I$ of size~$\overline{k}$,
the CCAV score of the $(m-\overline{k})$-committee $N\setminus I$ is~$\R$. Conversely, if there is an $(m-\overline{k})$-committee~$\ww$ of CCAV score~$\R$,
every vote has at least one of its approved candidates in the committee. Due to the construction, this implies that $N\setminus \ww$ is an independent set.
\end{proof}

For PAV, we also obtain a {\wahns} result via a reduction from the following problem.

\EPP
{Minimum ($\kappa$)-Vertex Subgraph ($\kappa$-MVS)}
{A graph $G=(N, A)$ and two integers~$\kappa$ and~$\ell$.}
{$\kappa$.}
{Is there $S\subseteq N$ such that $\abs{S}=\kappa$ and $G[N\setminus S]$ has at most~$\ell$ edges?}

Cai~\cite{DBLP:journals/cj/Cai08} proved that the $\kappa$-{\sc{MVS}} problem is {\wah} even when the input graph is regular, i.e., all vertices have the same degree.

\begin{theorem}
\label{thm-pav-wah-dual-k}
{\wind{PAV}} is {\memph\wah} with respect to~$\overline{k}$, even when every vote approves two candidates.
\end{theorem}

\begin{proof}
We provide a reduction from $\kappa$-{\sc{MVS}} to {\wind{PAV}}. Let $(G, \kappa, \ell)$ be an instance of $\kappa$-{\sc{MVS}} such that every vertex of~$G$ has degree~$r$ for some positive integer~$r$. Let $G=(N, A)$ and let $m=\abs{N}$. Without loss of generality, we assume that $\kappa<m$.
We construct an instance $((C, V), k, \R)$ of {\wind{PAV}} as follows. For each vertex~$c\in N$, we create one candidate in~$C$ denoted by the same symbol,
and for each edge $\edge{c}{c'}\in A$, we create one vote approving~$c$ and~$c'$ in~$V$.
We set $k=m-\kappa$, and hence $\overline{k}=\kappa$. Finally, we set $\R=(m-\kappa)\cdot r-\frac{\ell}{2}$. It remains to show the correctness.

$(\Rightarrow)$ Suppose that there exists $S\subseteq N$ so that $\abs{S}=\kappa$ and~$G[N\setminus S]$ contains exactly $\ell'\leq \ell$ edges.
Then the PAV score of the committee $N\setminus S$ is \[\frac{3}{2}\ell'+\left((m-\kappa)\cdot r-2\ell'\right)= (m-\kappa)\cdot r-\frac{\ell'}{2}\geq \R.\]
Obviously, $\abs{N\setminus S}=m-\abs{S}=k$. Therefore, the {\wind{PAV}} instance is a {\yesins}.

$(\Leftarrow)$ Conversely, suppose that $\ww\subseteq C$ is an $(m-\kappa)$-committee of PAV score at least~$\R$. Let~$\ell'$ be the number of votes whose both approved candidates are in~$\ww$.
Due to the reduction, every candidate is approved by exactly~$r$ votes.
Therefore, there are exactly $(m-\kappa)\cdot r-2 \ell'$ votes which have exactly one of their approved candidates in~$\ww$.
Hence, the PAV score of~$\ww$ is $\frac{3}{2}\ell'+((m-\kappa)\cdot r-2\ell')$, which is at least~$\R$.
It immediately follows that $\ell'\leq \ell$, implying that $N\setminus \ww$ is a {\yes}-witness of the {\prob{$\kappa$-MVS}} instance $(G, \kappa, \ell)$.
\end{proof}

It should be pointed out that a reduction established by Skowron, Faliszewski, and Lang~\cite[Theorem~5]{DBLP:journals/ai/SkowronFL16} for showing {\nphns} of a related problem also implies that {\wind{PAV}} is {\wah} with respect to~$\overline{k}$.\footnote{Their reduction is from {\prob{Vertex Cover}} restricted to $3$-regular graphs, and a slight modification results in a reduction based on {\prob{Vertex Cover}} restricted to regular graphs which is {\wah} parameterized by~$m-\kappa$~\cite{DBLP:journals/tcs/DowneyF95}.}

Now we move on to the parameter~$\R$. We observe that {\wind{CCAV}} is equivalent to the {\prob{Partial Hitting Set}} problem which has been intensively studied in the literature.

\EPP
{Partial Hitting Set}
{A universe~$U$, a collection~$\mathcal{S}$ of subsets of~$U$, and two nonnegative integers~$a$ and~$b$.}
{$b$.}
{Is there $S\subseteq U$ such that $\abs{S}=a$ and~$S$ intersects at least~$b$ elements of~$\mathcal{S}$?}

Clearly, by taking $C=U$, $V=\mathcal{S}$, $k=a$, and $\R=b$ in the above definition we obtain {\wind{CCAV}}.
Bl\"{a}ser~\cite{DBLP:journals/ipl/Blaser03} derived an algorithm running in time $\bigos{2^{\bigo{b}}}$ for {\sc{Partial Hitting Set}}.
This leads to the following corollary.

\begin{corollary}[\cite{DBLP:journals/ipl/Blaser03}]
\label{thm-ccav-R-FPT}
{\wind{CCAV}} can be solved in time $\bigos{2^{\bigo{\R}}}$.
\end{corollary}

Observe that every {\wind{CCAV}} instance where $\R>n$ is a {\noins}. As a result, Corollary~\ref{thm-ccav-R-FPT} implies an algorithm for {\wind{CCAV}} running in time~$\bigos{2^{\bigo{n}}}$,
which appreciably improves the~$\bigos{n^n}$ algorithm presented in~\cite{DBLP:journals/jair/BetzlerSU13}.

\begin{corollary}[\cite{DBLP:journals/ipl/Blaser03}]
\label{cor-ccav-fpt-n}
{\wind{CCAV}} can be solved in time~$\bigos{2^{\bigo{n}}}$.
\end{corollary}

Bl\"{a}ser~\cite{DBLP:journals/ipl/Blaser03} also presented a randomized algorithm of running time~$\bigos{(2e)^b}$ which
solves the  {\sc{Partial Hitting Set}} problem correctly with probability at least $1-e^{-1}$, where~$e$ is the Napier's constant.
We arrive at the following corollary.

\begin{corollary}[\cite{DBLP:journals/ipl/Blaser03}]
\label{cor-ccav-randomized}
There is a randomized algorithm running in time~$\bigos{(2e)^d}$ which solves {\wind{CCAV}} correctly with probability at least $1-e^{-1}$.
\end{corollary}

\subsection{Parameters~\texorpdfstring{$\triangle_{\text{V}}$}{Lg} and~\texorpdfstring{$\triangle_{\text{C}}$}{Lg}}
In this section, we study the two parameters $\triangle_{\text{V}}$ and~$\triangle_{\text{C}}$.
First, as {\sc{Vertex Cover}} remains {\nph} when restricted to cubic graphs (i.e.,~$3$-regular graphs)~\cite{gareyJS1976}, the reduction for the {\nphns} of {\wind{MAV}} established by LeGrand~\cite{LeGrand2004Tereport} implies that {\wind{MAV}} is {\nph} even when every candidate is approved by three votes and every vote approves  two candidates.
It remains to study the cases where $\triangle_{\text{V}}\leq 1$  or $\triangle_{\text{C}}\leq 2$.

Recall that for a committee selection rule~$\tau$ and an election $E=(C, V)$, ${\tau}(E, k)$ is the set of all optimal~$k$-committees under~$\tau$. More precisely, for~$\tau$ being MAV, $\tau(E, k)$ consists of all $k$-committees of~$C$ with the minimum Hamming distance to the votes, and for~$\tau$ being  CCAV and PAV, ${\tau}(E, k)$ consists of all $k$-committees of~$C$ with the maximum CCAV and PAV scores, respectively.

It is easy to see that every election~$E=(C, V)$ where every vote approves only one candidate admits a $k$-committee, $k\leq \abs{C}$, which is optimal with respect to MAV, CCAV, and PAV at the same time. In fact, any optimal $k$-committee under the approval voting (AV) is such a $k$-committee. The AV score of a candidate~$c$ is~$\abs{V(c)}$, the number of votes approving~$c$, and AV selects  $k$-committees with the maximum sum of AV scores of candidates in the committees.

\begin{observation}
\label{lem-relation-when-triangle-c-1}
Let $E=(C, V)$ be an election where every vote in~$V$ approves at most one candidate, and let $k\leq \abs{C}$ be an integer.
It holds that {\memph{\text{AV}}}$(E, k)\subseteq {\memph{\text{MAV}}}(E,k)$ and ${\memph{\text{AV}}}(E, k)={\memph{\text{CCAV}}}(E,k)={\memph{\text{PAV}}}(E,k)$.
\end{observation}

To see that ${\text{AV}}(E, k)\subseteq {\text{MAV}}(E,k)$ instead of  $\text{AV}(E, k)= {\text{MAV}}(E,k)$ in Observation~\ref{lem-relation-when-triangle-c-1}, consider an election with two candidates~$a$ and~$b$, where~$a$ is approved by two votes, and~$b$ is approved by a third vote which does not approve~$a$. For $k=1$, both~$\{a\}$ and~$\{b\}$ are optimal with respect to MAV, but only~$\{a\}$ is optimal with respect to AV.

The following corollary follows from Observation~\ref{lem-relation-when-triangle-c-1} and the clear fact that an optimal $k$-committee with respect to AV can be computed in polynomial time.

\begin{corollary}
\label{cor-poly-every-voter-approves-one-candidate}
For each $\tau\in \{{\emph{\text{MAV}}}, {\emph{\text{CCAV}}}, {\emph{\text{PAV}}}\}$, {\wind{$\tau$}} is polynomial-time solvable when $\triangle_{\emph{\text{V}}}\leq 1$.
\end{corollary}

Regarding $\triangle_{\text{C}}$, we have several polynomial-time solvability results for $\triangle_{\text{C}}\leq 2$. One of our results is based on the following polynomial-time solvable problem~\cite{DBLP:conf/icaart/Lin11}.

\EP
{Simple $b$-Edge Cover of Multigraphs (SECM)}
{A multigraph $G=(N, A)$, a function $f: N\rightarrow \mathbb{Z}^+$, and an integer~$\kappa$.}
{Is there $A'\subseteq A$ such that $\abs{A'}\leq \kappa$ and every $v\in N$ is incident to at least~$f(v)$ edges in~$A'$?}

If we require $\abs{A'}=\kappa$ in the above definition, we obtain the exact version of {\prob{SECM}} (\prob{E-SECM}).
Clearly, {\prob{E-SECM}} can also be  solved in polynomial time.

Every election $E=(C, V)$ can be represented by a multihypergraph where every vote $v\in V$ is considered as a vertex and every candidate $c\in C$ is considered as an edge consisting of  vertices in~$V(c)$. When a candidate is approved by only one vote, there is a loop on this vote. We use~$H(E)$ to denote this multihypergraph representing~$E$. Clearly, given an election, its multihypergraph representation can be computed in polynomial time. When $\triangle_{\text{C}}\leq 2$,~$H(E)$ degenerates to a multigraph.

For a class $\mathcal{H}$ of multihypergraphs, we say that an election~$E$ is an {\it{$\mathcal{H}$-election}} if $H(E)\in \mathcal{H}$.

\begin{theorem}
\label{thm-maximin-poly-cases}
{\wind{MAV}} and {\wind{CCAV}} are polynomial-time solvable if $\triangle_{\emph{\text{C}}}\leq 2$.
\end{theorem}

\begin{proof}
We derive polynomial-time algorithms for the special cases of {\wind{MAV}} and {\wind{CCAV}} stated in the theorem as follows. Let $(E, k, \R)$ be an instance of {\wind{MAV}} or  {\wind{CCAV}}, where $E=(C, V)$. We first compute the multigraph representation~$H(E)=(V, A)$ of~$E$ which can be done in polynomial time, where~$A$ is the set of edges corresponding to candidates in~$C$.
\medskip

{\noindent\bf{{{MAV}}.}} We solve {\wind{MAV}} by reducing it to {\prob{E-SECM}}.
Recall that for each $v\in V$,~$\abs{v}$ is the number of candidates approved by~$v$ which is equal to the number of edges incident to~$v$ in~$H(E)$.
Observe that the Hamming distance between every vote~$v$ and every $k$-committee can be at most $\abs{v}+k$. Therefore, if $\R\geq \abs{v}+k$, we can safely remove~$v$ from~$V$ without changing the answer to the instance. In light of this fact, we assume now that $\abs{v}+k>\R$ for every $v\in V$.
Let~$f: V \rightarrow \mathbb{Z}^+$ be a function such that $f(v)=\left\lceil\frac{\abs{v}+k-\R}{2}\right\rceil$ for every $v\in V$.
By setting $\kappa=k$ we complete the construction of an instance $(H(E), f, \kappa)$ of {\prob{E-SECM}}.
Assume that there exists $A'\subseteq A$ of cardinality~$\kappa$ so that every $v\in V$ is incident to at least~$f(v)$ edges in~$A'$. Let~$\ww$ be the~$k$-committee corresponding to~$A'$.
The Hamming distance between every vote~$v\in V$ and~$\ww$ is \[\abs{v\setminus \ww}+\abs{\ww\setminus v}\leq (\abs{v}-f(v))+(k-f(v))\leq \R.\] The proof for the other direction is analogous.
\medskip

{\noindent\bf{{{CCAV}}.}} We derive a greedy algorithm for {\wind{CCAV}}.
Let~$H'$ be the graph obtained from~$H(E)$ by
\begin{enumerate}
    \item[(1)] removing all loops, and
    \item[(2)] for every two vertices between which there are multiple edges, removing all but any arbitrary one of these multiple edges.
\end{enumerate}
Let~$M$ be a maximum matching of~$H'$, and let~$V(M)$ be the set of vertices saturated by~$M$. We distinguish between two cases.

\begin{itemize}
\item If $\abs{M}\geq k$, we arbitrarily select~$k$ edges in~$M$, and let~$\ww$ be the $k$-committee corresponding to these selected edges.

\item If $\abs{M}<k$, let~$\ww$ be the set of candidates corresponding to edges in~$M$.
Let~$E'$ be the election obtained from~$E$ by removing all votes in~$V(M)$ and all candidates approved only by votes in~$V(M)$.
As~$M$ is a maximum matching of~$H'$, no two votes of~$E'$ approve a common candidate. Then, for every nonempty vote in~$E'$, we arbitrarily select one candidate approved by the vote. Let~$\ww'$ denote the set of all these selected candidates. If $\abs{\ww'}\geq k-\abs{M}$, we include into~$\ww$ any arbitrary $k-\abs{M}$ candidates from~$\ww'$; otherwise, we include into~$\ww$ all candidates of~$\ww'$ together with any arbitrary $k-\abs{\ww\cup \ww'}$ remaining candidates.
\end{itemize}
In either case, we conclude that the given instance of {\wind{CCAV}} is a {\yesins} if and only if $\textsf{CCAV}(V, \ww)\geq \R$.
\end{proof}

Now we study special cases of {\wind{PAV}} where~$\triangle_{\text{C}}$ and~$\triangle_{\text{V}}$ are very small integers. We need the following notions.
A {\emph{path}} is a graph comprised of a sequence~$v_1$,~$v_2$,~$\dots$,~$v_t$ of~$t$ vertices and~$t-1$ edges so that there is an edge between two vertices if and only if they are consecutive in the sequence. A {\emph{cycle}} is a graph obtained from a path by adding an edge between the first and the last vertices. A {\emph{hairstick}} is a graph obtained from a path by adding one loop either on the first vertex or on the last vertex. A {\emph{double headed hairstick}} (DH-hairstick) is a graph obtained from a path by adding one loop on both the first vertex and  on the last vertex. We refer to  Figure~\ref{fig-small} for an illustration of these graphs.

\begin{figure}
\centering
{\includegraphics[width=\textwidth]{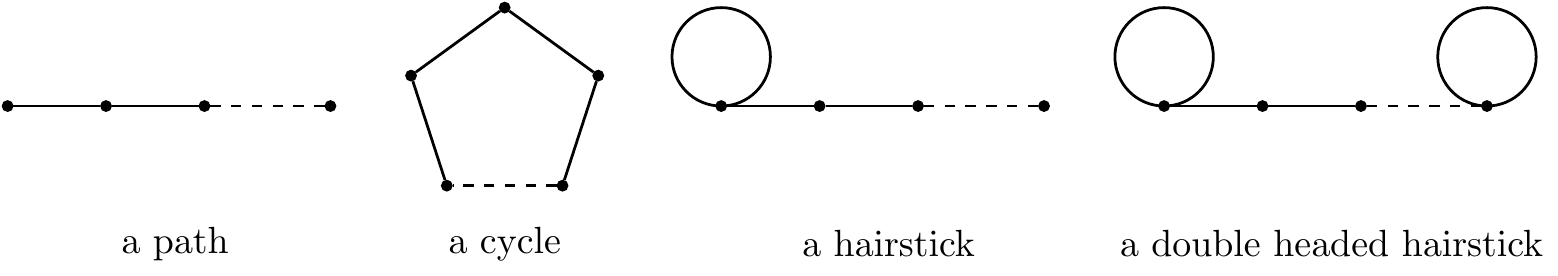}
}
\caption{An illustration of cycles, paths, hairsticks, and DH-hairsticks.}
\label{fig-small}
\end{figure}

\begin{lemma}
\label{lem-trivial-cases-pav}
Let $\mathcal{H}$ be the set of all paths, cycles, hairsticks, and DH-hairsticks. Then, given an $\mathcal{H}$-election, we can compute a PAV optimal $k$-committee of the election in polynomial time.
\end{lemma}

We defer the proof of Lemma~\ref{lem-trivial-cases-pav} to the appendix.

\begin{theorem}
\label{thm-pav-p-special}
{\wind{PAV}} is polynomial-time solvable if $\triangle_{\emph{\text{C}}}\leq 1$ or $\triangle_{\emph{\text{V}}}=\triangle_{\emph{\text{C}}}= 2$.
\end{theorem}

\begin{proof}
Let $I=(E, k, \R)$ be an instance of {\wind{PAV}} where $E=(C, V)$.

We consider first the case where $\triangle_{\text{C}}\leq 1$.
The following algorithm finds an optimal $k$-committee. First, let $\ww=\emptyset$.
Then, we arrange votes in~$V$ in a cyclic order (the relative orders of the votes do not matter), and starting from any arbitrary vote we consider the votes one-by-one in a clockwise order. In particular, if the currently considered vote~$v$ has at least one approved candidate not contained in~$\ww$, i.e., $v\setminus \ww \neq \emptyset$, we add any arbitrary candidate from $v\setminus \ww$ into~$\ww$; otherwise we proceed to the next vote. The procedure runs until $\abs{\ww}=k$ or~$\ww$ cannot be expanded in the way described above. Finally, we conclude that~$I$ is a {\yesins} if and only if $\textsf{PAV}(V, \ww)\geq \R$.

Now we consider the case where $\triangle_{\text{C}}=\triangle_{\text{V}}=2$.
We first compute the multigraph representation~$H(E)$ of~$E$. Given a subgraph of~$H(E)$, the subelection of~$E$ restricted to this subgraph refers to~$E$ restricted to candidates and votes corresponding respectively to the edges and the vertices of the subgraph. Observe that every connected component of~$H(E)$ is either a path, a cycle, a hairstick, or a DH-hairstick, and by Lemma~\ref{lem-trivial-cases-pav}, for every integer~$j$, an optimal $j$-committee of each subelection restricted to a connected component can be computed in polynomial time. Based on this, we derive a dynamic programming algorithm. Precisely, let $(H_1, H_2, \dots, H_z)$ be an arbitrary order of the connected components of~$H(E)$. For each $i\in [z]$, let $m^{\leq i}$ be the number of edges in the first~$i$ connected components in the order. We maintain a table~$T(i, j)$ where $i\in [z]$ and $j \leq \min\{k, m^{\leq i}\}$ is a nonnegative integer. We define~$T(i, j)$ as the PAV score of an optimal $j$-committee in the election restricted to the first~$i$ connected components. By Lemma~\ref{lem-trivial-cases-pav}, $T(1, j)$ for all possible~$j$ can be computed in polynomial time. We use the following recursion to compute~$T(i, j)$, assuming that all entries $T(i', j')$ such that $i'<i$ have been computed. For each $j'\leq j$, let~$\R(j')$ be the PAV score of an optimal $j'$-committee of~$E$ restricted to~$H_i$, which can be computed in polynomial time by Lemma~\ref{lem-trivial-cases-pav}. Then, we have that
\[T(i, j)=\max_{\substack{j'\in [j]\cup \{0\},\\ j-j'\leq m^{\leq i-1}}} \{\R(j')+T(i-1, j-j')\}.\]

After all entries are computed, we conclude that the given instance~$I$ is a {\yesins} if and only if $T(z, k)\geq \R$.
\end{proof}

\section{Combined Parameters}
\label{sec-combined-parameter}
The study in the previous section revealed that the parameters~$k$,~$\overline{k}$,~$\triangle_{\text{V}}$, and~$\triangle_{\text{C}}$ generally lead
to fixed-parameter intractability results.
In this section, we explore whether combining two or three of them offers us {\fpt}-results.

\subsection{Combining~\texorpdfstring{$k$}{Lg} with~\texorpdfstring{$\triangle_{\text{V}}$}{Lg} and~\texorpdfstring{$\triangle_{\text{C}}$}{Lg}}
We first consider the combined parameter $k+\triangle_{\text{C}}$, starting with an {\fpt}-result for {\wind{MAV}}.

\begin{theorem}
\label{thm-mav-fpt-k-triangle-c}
{\wind{MAV}} is {\memph\fpt} with respect to $k+\triangle_{\memph{\text{C}}}$.
\end{theorem}

\begin{proof}
We prove the theorem by giving an {\fpt}-algorithm for {\wind{MAV}} as follows. Let $I=(E, k, \R)$ be an instance of {\wind{MAV}}, where $E=(C, V)$.
Let $n=\abs{V}$. If $n\leq k\cdot\triangle_{\text{C}}+1$, then as {\wind{MAV}} is {\fpt} with respect to~$n$~\cite{DBLP:conf/atal/MisraNS15}, we can solve~$I$ in {\fpt}-time in $k+\triangle_{\text{C}}$.
Otherwise, let $(v_1,v_2,\dots,v_n)$ be a linear order on~$V$ such that $\abs{v_i}\geq \abs{v_{i+1}}$ for every $i\in [n-1]$, i.e.,~$v_i$ approves at least the same number of candidates as~$v_{i+1}$ does.
Then, the algorithm deletes the last $n-(k\cdot \triangle_{\text{C}}+1)$ votes in this order,
and solves the remaining instance by an {\fpt}-algorithm with respect to the number of votes (e.g., the one presented in~\cite{DBLP:conf/atal/MisraNS15}).

It remains to prove the correctness of the algorithm.
Let~$E'$ be the election after the deletion of the votes as described above. Clearly, every~$k$-committee of the original election with MAV score~$\R$ is a $k$-committee of~$E'$ with MAV score at most~$\R$.
To show the correctness for the opposite direction, let~$\ww$ be a~$k$-committee of~$E'$ with MAV score~$\R$.
As every candidate is approved by at most~$\triangle_{\text{C}}$ votes, at most~$k\cdot \triangle_{\text{C}}$ votes in~$E'$ intersect~$\ww$.
As~$E'$ contains the first $k\cdot \triangle_{\text{C}}+1$ votes in the order defined above, there exists $i\in [k\cdot \triangle_{\text{C}}+1]$ so that the vote~$v_i$ does not approve any candidate from~$\ww$.
The Hamming distance between~$\ww$ and~$v_i$ is $k+\abs{v_i}\leq \R$. As $\abs{v_j}\leq \abs{v_i}$ for every deleted vote~$v_j$, $j\geq k\cdot \triangle_{\text{C}}+2$,
the Hamming distance between~$\ww$ and~$v_j$ is at most $k+\abs{v_j}\leq k+\abs{v_i}\leq \R$. It follows that~$\ww$ has MAV score~$\R$ in the election~$E$.
\end{proof}

Let us move on to CCAV.
Obviously, if~$\ww$ is a~$k$-committee, then at most $k\cdot \triangle_{\text{C}}$ votes
intersect~$\ww$. Hence, the CCAV score of every optimal~$k$-committee is bounded from above by $k\cdot \triangle_{\text{C}}$. This observation leads to a simple algorithm for {\wind{CCAV}}: if $\R>k\cdot \triangle_{\text{C}}$, return ``{\no}''; otherwise, solve it by Corollary~\ref{thm-ccav-R-FPT} in time $\bigos{2^{\bigo{\R}}}=\bigos{2^{\bigo{k\cdot \triangle_{\text{C}}}}}$.

\begin{corollary}
\label{thm-ccav-fpt-k-triagnel-c}
{\wind{CCAV}} is {\memph\fpt} with respect to $k+\triangle_{\memph{\text{C}}}$.
\end{corollary}

Now we study the combined parameter $k+\triangle_{\text{V}}$.
Observe that {\wind{MAV}} admits a straightforward {\fpt}-algorithm with respect to $k+\triangle_{\text{V}}$: if $\R>k+\triangle_{\text{V}}$, every~$k$-committee has MAV score at most~$\R$, and hence we directly conclude that the given instance is a {\yesins}; otherwise, we solve it by the {\fpt}-algorithm with respect to~$\R$ proposed in~\cite{DBLP:conf/atal/MisraNS15}.

\begin{corollary}
\label{thm-mav-fpt-k-triangle-v}
{\wind{MAV}} is {\memph\fpt} with respect to $k+\triangle_{\memph{\text{V}}}$.
\end{corollary}

Aziz~et~al.~\cite{DBLP:conf/atal/AzizGGMMW15} proved that {\wind{PAV}} is {\wah} when the parameter is~$k$ even if $\triangle_{\text{V}}=2$.
This implies that {\wind{PAV}} is {\wah} with respect to $k+\triangle_{\text{V}}$.
For CCAV, Betzler, Slinko, and Uhlmann~\cite{DBLP:journals/jair/BetzlerSU13} proved that {\wind{CCAV}} is {\wbh} with respect to~$k$ by a reduction from the {\sc{Hitting Set}} problem.  However, in the reduction the maximum number of candidates approved by a vote is not bounded from above by a constant.
Using a reduction from the {\sc{Partial Vertex Cover}} problem, we show that {\wind{CCAV}} is {\wah} with respect to~$k$ in this special case.

\EPP
{Partial Vertex Cover (PVC)}
{A graph~$G$ and two integers~$\kappa$ and~$\ell$.}
{$\kappa$.}
{Is there a subset $S\subseteq V(G)$ such that $\abs{S}=\kappa$ and~$S$ covers at least~$\ell$ edges of~$G$?}

It is known that PVC is {\wah} with respect to~$\kappa$~\cite{DBLP:conf/wads/GuoNW05}.

\begin{theorem}
\label{thm-ccav-wah-k-voter-approve-two-candidates}
{\wind{CCAV}} is {\memph\wah} with respect to~$k$ even if every vote approves two candidates.
\end{theorem}

\begin{proof}
Given an instance $(G, \kappa, \ell)$ of {PVC}, we create a {\wind{CCAV}} instance as follows. We regard each vertex as a candidate and regard each edge as a vote approving exactly the two candidates corresponding to its two endpoints. The reduction is completed by setting $k=\kappa$ and $\R=\ell$. The correctness is easy to see.
\end{proof}

\subsection{Combining~\texorpdfstring{$\overline{k}$}{Lg} with~\texorpdfstring{$\triangle_{\text{V}}$}{Lg} and~\texorpdfstring{$\triangle_{\text{C}}$}{Lg}}

We have shown that {\wind{MAV}}, {\wind{CCAV}}, and {\wind{PAV}} are {\wah} or {\wbh} with respect to the single parameter~$\overline{k}$ even when $\triangle_{\text{V}}=2$.
It follows that these problems are {\wah} or {\wbh} when parameterized by $\overline{k}+\triangle_{\text{V}}$.
Hence, we focus only on the combined parameter $\overline{k}+\triangle_{\text{C}}$.
We first prove that {\wind{MAV}} is {\fpt} with respect to this parameter by reducing it to an {\fpt} problem which is a generalization of the $r$-{\sc{Set Packing}} problem.

\EPP
{Generalized $r$-Set Packing (G$r$SP)}
{A universe~$U$, a multiset~$\mathcal{S}$ of $r$-subsets of~$U$,  a function $f: U\rightarrow \mathbb{N}_{0}$, and an integer~$\kappa$.}
{$\kappa+r$.}
{Is there an $S\subseteq \mathcal{S}$ such that $\abs{S}=\kappa$ and every $u\in U$ occurs in at most~$f(u)$ elements of~$S$?}

It is known that {\prob{G$r$SP}} is {\fpt}~\cite{DBLP:journals/jcss/YangG17}\footnote{The definition of G$r$SP in~\cite{DBLP:journals/jcss/YangG17} requests~$f(u)$ to be positive for all $u\in U$. However, it is fairly easy to see that if we allow $f(u)=0$ for some $u\in U$, we can safely remove all $X \in \mathcal{S}$ such that~$u\in X$ from the instance without changing the answer to the instance. So, allowing $f(u)=0$ for $u\in U$ does not destroy the fixed-parameter tractability of the problem.}.
It is easy to verify that the variant of the {\prob{G$r$SP}} problem where each $s\in \mathcal{S}$ is of cardinality at most~$r$, instead of exactly~$r$, is reducible to {\prob{G$r$SP}} in polynomial time:
for each $s\in \mathcal{S}$ such that $\abs{s}<r$, create $r-\abs{s}$ new elements in~$U$, add them into~$s$, and set $f(u)=1$ for every newly introduced element in~$U$.
Therefore, the above variant is also {\fpt}. In the following, we use {\prob{G$r_{\leq}$SP}} to denote this variant.

\begin{theorem}
\label{thm-mav-fpt-bar-k-triangle-c}
{\wind{MAV}} is {\memph\fpt} with respect to $\overline{k}+\triangle_{\memph{\text{C}}}$.
\end{theorem}

\begin{proof}
Let $I=(E, k, \R)$ be an instance of {\wind{MAV}}, where $E=(C, V)$. Note that the Hamming distance between each vote~$v$ and each $k$-committee is at least $\abs{k-\abs{v}}$. Therefore, if there exists a vote $v\in V$ such that $\abs{v}<k$ and $d<k-\abs{v}$, we immediately conclude that the given instance~$I$ is a {\noins}. In what follows, we assume that for all $v\in V$ either it holds that $\abs{v}\geq k$, or it holds that $\abs{v}<k$ and $d\geq k-\abs{v}$.
We reduce~$I$ to an instance of {\prob{G$r_{\leq}$SP}} in polynomial time as follows.
Precisely, let $U=V$, and let $\mathcal{S}=\{V(c) \setmid c\in C\}$.
Clearly, each element of~$\mathcal{S}$ is
of cardinality at most~$\triangle_{\text{C}}$. Hence, we set $r=\triangle_{\text{C}}$. Regarding the function~$f$, for each $v\in V$ we define
\[f(v)=\left\lfloor \frac{\R+\abs{v}-k}{2}\right\rfloor.\] By the above assumption, the function~$f$ is nonnegative. Finally, we define $\kappa=\overline{k}=\abs{C}-k$.

It remains to prove that the two instances are equivalent.

$(\Rightarrow)$ Suppose that the {\prob{G$r_{\leq}$SP}} instance is a {\yesins}. In particular, let $S\subseteq \mathcal{S}$ be such that $\abs{S}=\kappa$ and every $u\in U$ occurs in at most~$f(u)$ elements of~$S$. Let~$C'$ be the subset of candidates corresponding to~$S$. Clearly, $\abs{C'}=\abs{S}=\overline{k}$. We claim that $\ww=C\setminus C'$ is a~$k$-committee with MAV score at most~$\R$.
Due to the above construction, every vote~$v$ occurs in at most~$f(v)$ submultisets from~$S$. This implies that~$v$ has at least $\abs{v}-f(v)$ of its approved candidates in~$\ww$.
Hence, the Hamming distance between~$v$ and~$\ww$ is at most \[f(v)+(k-(\abs{v}-f(v)))=2\cdot f(v)+k-\abs{v}\leq \R.\] In other words, the instance~$I$ is a {\yesins}.

$(\Leftarrow)$ Suppose that~$I$ is a {\yesins}, i.e., there is a~$k$-committee~$\ww\subseteq C$ such that $\textsf{MAV}(V, \ww)\leq \R$. Let $C'=C\setminus \ww$ and let  $S=\{V(c) \setmid c\in C'\}$.
Then, for each vote $v\in V$ at least $\left\lceil \frac{\abs{v}+k-\R}{2}\right\rceil$ of its approved candidates must be in~$\ww$.
In other words, at most $\abs{v}-\left\lceil \frac{\abs{v}+k-\R}{2}\right\rceil=f(v)$ of $v$'s approved candidates can be in~$S$.
This implies that the {\prob{G$r_{\leq}$SP}} instance is a {\yesins}.
\end{proof}

Next, we prove that {\wind{CCAV}} is {\fpt} with respect to the combined parameter $\overline{k}+\triangle_{\text{C}}$ as well.

\begin{theorem}
\label{thm-ccav-wd-fpt-bar-k-c}
{\wind{CCAV}} is {\memph\fpt} with respect to $\overline{k}+\triangle_{\memph{\text{C}}}$.
\end{theorem}

\begin{proof}
We prove Theorem~\ref{thm-ccav-wd-fpt-bar-k-c} by deriving a branch-and-bound {\fpt}-algorithm with respect to $\overline{k}+\triangle_{\text{C}}$.
Let $I=((C, V), k, \R)$ be an instance of {\wind{CCAV}}.
Throughout the algorithm, let $\overline{k}=\abs{C}-k$.

We first remove from~$C$ all candidates not approved by any votes, and remove from~$V$ all empty votes. Then, we assume that neither of~$C$ and~$V$ is empty, since otherwise the instance can be solved trivially.
Observe now that if a vote~$v\in V$ approves more than~$\overline{k}$ candidates, any~$k$-committee intersects~$v$.
In light of this observation, we need only to focus on votes approving at most~$\overline{k}$ candidates.
Let~$U\subseteq V$ be the submultiset of votes from~$V$ approving at most~$\overline{k}$ candidates. Let $B=\cup_{v\in U}v$ be the set of candidates approved by votes from~$U$. If $U=\emptyset$ or $\abs{B}\leq k$, any $k$-committee containing~$B$ satisfies all votes. In this case, we conclude that~$I$ is a {\yesins} if and only if $\R\leq \abs{V}$. Assume now that~$\abs{B}\geq k+1$.
For each candidate $c\in C$, let~$V^1(c)$ be the multiset of votes in~$U$ that approve~$c$ only, i.e., $V^1(c)=\{v\in U \setmid v=\{c\}\}$.
Let $c^{\star}\in B$ be such that  $\abs{V^1(c^{\star})} \leq \abs{V^1(c)}$
for all $c\in B$. Let $A=\bigcup_{v\in V(c)\cap U} v$. A significant observation is that there is an optimal~$k$-committee which does not contain~$A$. (A proof for the observation is given later.)
In line with this observation, we solve the instance by branching on which candidate from~$A$ is not contained in a certain optimal~$k$-committee.
In particular, we create~$\abs{A}$ branching cases, one for each $x\in A$. In the branching case for $x\in A$, we reset $C\coloneqq C\setminus \{x\}$ (which implies that~$\overline{k}$ is decreased by one), $\R \coloneqq \R-\abs{V\setminus U}$, and $V \coloneqq U$, and solve the subinstance iteratively by the above procedure.
We use two pruning criteria to determine when to terminate the branching: $\R\leq 0$ or $\overline{k}=0$. When we arrive at a branching node where $\R\leq 0$ and $\overline{k}\geq 0$, it holds clearly that the given instance~$I$ is a {\yesins}, and thus in this case we
 terminate the whole algorithm by returning~``\yes''. When we arrive at a branching node where $\overline{k}=0$ and $\R>0$, we determine that~$I$ is a {\yesins} if $\textsf{CCAV}(V, C)\geq \R$, and discard the corresponding branching case otherwise.

Now we prove the correctness of the above mentioned observation.
To this end, assume that~$\ww$ is an optimal~$k$-committee such that $A\subseteq \ww$. Then, we can obtain another optimal~$k$-committee from~$\ww$ by replacing~$c^{\star}$ (notice that $c^{\star}\in A$) with any
arbitrary candidate from~$B\setminus \ww$ (as $\abs{B}\geq k+1$ such a candidate exists). The reason is that every vote from $(V(c^{\star})\cap U)\setminus V^1(c^{\star})$ approves at least one candidate from $\ww\setminus \{c^{\star}\}$ (because $A\subseteq \ww$). Hence, removing~$c^{\star}$ from~$\ww$ may only affect the existence of representatives of votes from~$V^1(c^{\star})$ in the committee. However, as $\abs{V^1(c^{\star})} \leq \abs{V^1(c)}$ for all $c\in B$, adding any arbitrary candidate from~$B\setminus \ww$ into~$\ww$ makes all votes from~$V^1(c^{\star})$ which do not have any representatives before have representatives now.

Regarding the running time of the algorithm, note that $\abs{A}\leq \triangle_{\text{C}} \cdot \overline{k}$.
Hence, each branching node
has at most $\triangle_{\text{C}} \cdot \overline{k}$ children, implying that the branching algorithm has running time $\bigos{(\triangle_{\text{C}} \cdot \overline{k})^{\overline{k}}}$.
\end{proof}

\subsection{Combining~\texorpdfstring{$\R$}{Lg} and~\texorpdfstring{$\triangle_{\text{V}}$}{Lg}}
It is known that  {\wind{MAV}} and {\wind{CCAV}} are both {\fpt} with respect to the parameter~$\R$.
However, it is unknown whether {\wind{PAV}} is {\fpt} with respect to~$\R$.
Even though we are unable to resolve this open question, we provide an {\fpt}-algorithm when combining~$\R$ and~$\triangle_{\text{V}}$.

\begin{theorem}
\label{thm-pav-fpt-r-v}
{\wind{PAV}} is {\memph\fpt} with respect to $\R+\triangle_{\memph{\text{V}}}$.
\end{theorem}

\begin{proof}
We prove the theorem by giving a branch-and-bound {\fpt}-algorithm.
Let $I=(E, k, \R)$ be an instance of {\wind{PAV}}, where $E=(C, V)$.

First, if~$C$ contains a candidate approved by at least~$\R$ votes in~$V$, any~$k$-committee including
this candidate is a desired committee. So, in this case, we directly conclude that~$I$ is a {\yesins}.

Second, let~$C_0\subseteq C$ be the set of candidates not approved by any votes in~$V$. We remove from the election $(C, V)$ all candidates in~$C_0$, and reset
$k \coloneqq \min \{k, \abs{C}\}$. It is easy to verify that the two instances before and after this step are equivalent.
Afterwards, if $k=\abs{C}$, we conclude that~$I$ is a {\yesins} if and only if $\textsf{PAV}(V, C)\geq \R$.

From now on, let us assume that every candidate is approved by at least one and at most~$\R-1$ votes in~$V$, and $k<\abs{C}$.
For each $S\subsetneq C$ and each $c\in C\setminus S$, we define
\[\Delta(S, c)=\pavscore{V}{S\cup \{c\}}-\pavscore{V}{S},\] as the marginal contribution of~$c$ to the PAV score of the committee~$S\cup \{c\}$.
In our branching tree, each branching node is associated with a subset $S\subsetneq C$ of cardinality at most~$k$, which is supposed to be contained in a desired committee.
The root of the branching tree is associated with $S=\emptyset$.
Suppose we are now at a branching node associated with a subset $S\subsetneq C$. Let $c\in C\setminus S$ be a candidate such that $\Delta(S, c)\geq \Delta(S, c')$ for all $c'\in C\setminus S$. Let $A=(\bigcup_{v\in V(c)}v) \setminus S$ be the set of candidates approved by votes approving~$c$ but not contained in~$S$.
Then, we create~$\abs{A}$ branching cases, one for each $x\in A$. The set associated with the branching case for $x\in A$ is $S\cup \{x\}$.
The correctness of our branching is rooted in the fact (a proof supporting this fact is provided later) that, if the given instance~$I$ is a {\yesins}, there exists at least one desired $k$-committee which includes at least one candidate from~$A$.

We terminate the branching when the branching depth of the current branching node reaches
$\min\{k, \R\cdot \triangle_{\text{V}}\}$. Recall that the branching depth of a branching node is the number of edges on the path from the root to the node in the branching tree. By our branching strategy, the branching depth of a branching node associated with~$S$ is exactly~$\abs{S}$. Therefore, when $k\leq \R\cdot \triangle_{\text{V}}$ and the branching depth of the current node associated with~$S$ reaches~$k$, we have that $\abs{S}=k$. In this case, if~$S$ has PAV score at least~$\R$ with respect to~$V$, we conclude that the given instance~$I$ is a {\yesins}. Consider now the case where
$k>\R\cdot \triangle_{\text{V}}$. Note that $\Delta(S,x)\geq \frac{1}{\triangle_{\text{V}}}$ for each possible~$S$ and~$x$. Therefore, if the current branching node associated with~$S$ has branching depth~$\R\cdot \triangle_{\text{V}}$,~$S$ is a $k'$-committee, $k'\leq k$, of PAV score at least~$\R$. As a result, when $k>\R\cdot \triangle_{\text{V}}$ and the current node has branching depth $\R\cdot \triangle_{\text{V}}$, we directly conclude that~$I$ is a {\yesins}. If none of the branching nodes leads to a ``\yes''-answer, we conclude that~$I$ is a {\noins}.

Note that $\abs{A}\leq \abs{V(c)}\cdot \triangle_{\text{V}}\leq \R\cdot \triangle_{\text{V}}$, where~$A$ is as defined above. Therefore, the running time of the whole algorithm is  bounded by $\bigos{(\R \cdot \triangle_{\text{V}})^{\R \cdot \triangle_{\text{V}}}}$.

To show the correctness, it suffices to prove the following claim (corresponding to the aforementioned fact).

\begin{claim}
\label{claim-a}
Let~$S\subsetneq C$ be a subset of at most $k-1$ candidates, and let $c\in C\setminus S$ be a candidate such that $\Delta(S, c)\geq \Delta(S, c')$ for all $c'\in C\setminus S$. Let $A=(\bigcup_{v\in V(c)}v)\setminus S$. Then, if there is a $k$-committee
$\ww\subsetneq C$ such that ${\emph{\textsf{PAV}}}(V, \ww)\geq \R$, $S\subsetneq \ww$, and $A\cap \ww=\emptyset$,
there exists a $k$-committee $\ww'\subsetneq C$ such that ${\emph{\textsf{PAV}}}(V, \ww')\geq \R$, $S\subseteq (\ww\cap \ww')$, and $A\cap \ww'\neq \emptyset$.
\end{claim}

Now we prove the claim. Let~$S$,~$c$, and~$A$ be as stipulated in the claim. Suppose that there is a $k$-committee
$\ww\subsetneq C$ such that ${\textsf{PAV}}(V, \ww)\geq \R$, $S\subsetneq \ww$, and $A\cap \ww=\emptyset$. Let~$c'$ be any candidate in $\ww\setminus S$, and let $\ww'=\ww\setminus \{c'\}\cup \{c\}$. Obviously, $S\subseteq (\ww\cap \ww')$ and $A\cap \ww'\neq \emptyset$. To complete the proof, it suffices to show that $\textsf{PAV}(V, \ww')\geq \R$.
First, as $\ww\cap A=\emptyset$ and $S\subseteq \ww$, it holds that $\Delta(\ww, c)=\Delta(S, c)$.
Second, as $S\subseteq \ww\setminus \{c'\}$, it holds that $\Delta(\ww\setminus \{c'\}, c')\leq \Delta(S, c')$. Third, it is clear that $\Delta(\ww\setminus \{c'\}, c)\geq \Delta(\ww, c)$. We also reiterate that $\Delta(S, c)\geq \Delta(S, c')$. Putting this all together, we obtain $\Delta(\ww\setminus \{c'\}, c)\geq \Delta(\ww\setminus \{c'\}, c')$ which implies $\textsf{PAV}(V, \ww')\geq {\textsf{PAV}}(V, \ww)\geq \R$. The proof is completed.
\end{proof}

We have derived many {\fpt}-algorithms with respect to different single parameters and their combinations.
Our results reveal that {\wind{MAV}} and {\wind{CCAV}} are {\fpt} with respect to any combinations of three single parameters studied so far.
Unfortunately, we have only a few {\fpt}-results for {\wind{PAV}} even for combinations of two single parameters.
We finish this section by remarking that {\wind{PAV}} is {\fpt} with respect to $k+\triangle_{\text{V}}+\triangle_{\text{C}}$: if $\R> {k\cdot \triangle_{\text{C}}}$, return ``\no''; otherwise, solve the instance in {\fpt}-time via Theorem~\ref{thm-pav-fpt-r-v}.

\begin{corollary}
\label{thm-pav-fpt-k-c-v}
{\wind{PAV}} is {\memph{\fpt}} with respect to $k+\triangle_{\memph{\text{V}}}+\triangle_{\emph{\text{C}}}$.
\end{corollary}



\section{Structural Parameters}
\label{sec-structural-parameter}
In this section, we study two structural parameters of incidence graphs of elections.
For an election~$E$, let~$G_E$ denote the incidence graph of~$E$.

\subsection{Treewidth}
Treewidth is a widely-studied notion to measure the closeness of a graph to a tree~\cite{DBLP:journals/jal/RobertsonS86}. It has been shown that a great deal of graphs stemming from innumerable combinatorial problems in a variety of areas have bounded treewidth (see, e.g.,~\cite{DBLP:journals/dam/KrauseLS20,DBLP:journals/almob/MarchandPB22,DBLP:journals/iandc/Thorup98,YamaguchiAM2003}). Additionally, it has been scrutinized  that even when treewidth is large, using tree decompositions could also be helpful for designing algorithms~\cite{DBLP:conf/icdt/ManiuSJ19}.
From a theoretical point of view, a myriad of {\nph} problems are known to be {\fpt} with respect to the parameter treewidth~\cite{DBLP:journals/tcs/CourcelleM93}. With regard to the applicability of treewidth in multiwinner voting, we point out that the treewidth of the incidence graph of an election is no greater than the number~$m$ of candidates and the number~$n$ of voters in the election, implying that any {\fpt}-algorithm for {\wind{$\tau$}} with respect to the treewidth also runs in {\fpt}-time in~$m$ and~$n$.

A {\it{tree decomposition}} of a graph $G=(N, A)$ is a tuple $(T,\mathcal{B})$, where~$T=(L, F)$ is a rooted tree with vertex set~$L$ and edge set~$F$,
and $\mathcal{B}=\{B_x \subseteq N \setmid x\in L\}$ is a
collection of subsets of vertices of~$G$ such that the following three conditions are fulfilled simultaneously:
\begin{itemize}
\item For each vertex $v\in N$ in~$G$ there exists at least one $B_x\in \mathcal{B}$ such that $v\in B_x$, i.e., every vertex of~$G$ is in at least one element of~$\mathcal{B}$.
\item For each edge $\edge{v}{u}\in A$ in~$G$ there exists at least one $B_x\in \mathcal{B}$ such that $v,u\in B_x$, i.e.,
every edge of~$G$ is contained in at least one element of~$\mathcal{B}$.
\item If a vertex $v\in N$ is in $B_x, B_y\in \mathcal{B}$, then~$v$ is in every $B_z\in \mathcal{B}$ where~$z$ is a vertex on the unique path between~$x$ and~$y$ in~$T$.
\end{itemize}
The {\it{width}} of the tree decomposition is defined as $\max_{B\in \mathcal{B}}{\abs{B}-1}$.
The {\it{treewidth}} of a graph~$G$, denoted~$\omega(G)$, is the minimum possible width of tree decompositions of~$G$.
Elements  of~$\mathcal{B}$ are called {\it{bags}}. To avoid confusion, in the following we call vertices of~$T$ {\it{nodes}}.

A more refined notion commonly used in designing {\fpt}-algorithms is the so-called nice tree decomposition~\cite{DBLP:conf/icalp/BodlaenderK91}.
In particular, a {\it{nice tree decomposition}} $(T, \mathcal{B})$ of a graph~$G$ is a tree decomposition of~$G$ which further satisfies the following conditions simultaneously:
\begin{itemize}
\item Every bag $B\in \mathcal{B}$ associated with the root or a leaf of~$T$ is empty, i.e., $B_x=\emptyset$ if~$x$ is the root or a leaf of~$T$.
\item Nonleaf nodes of~$T$ are categorized into {\it{introduce nodes, forget nodes}}, and {\it{join nodes}} such that:
\begin{itemize}
\item Each introduce node~$x$ has exactly one child~$y$ such that $B_y\subsetneq B_x$ and $\abs{B_x\setminus B_y}=1$, i.e.,~$B_x$ has exactly one more element than~$B_y$.
\item Each forget node~$x$ has exactly one child~$y$ such that $B_x\subsetneq B_y$ and $\abs{B_y\setminus B_x}=1$, i.e.,~$B_x$ is obtained from~$B_y$ by removing one element.
\item Each join node~$x$ has exactly two children~$y$ and~$z$ such that $B_x=B_y=B_z$.
\end{itemize}
\end{itemize}

It is easy to see from the definition that in a nice tree decomposition, each vertex can be introduced multiple times but can be only forgotten once. Moreover, when a vertex is introduced in a bag~$B_x$, all of its neighbors contained in bags associated with nodes in the subtree rooted at~$x$ are contained in~$B_x$.

\begin{lemma}[\cite{DBLP:books/sp/Kloks94}]
\label{lem-nice-tree-decomposition}
Let~$G$ be a graph of~$p$ vertices. Then, given a tree decomposition of~$G$ of width~$\omega$, a nice tree decomposition of~$G$ of width~$\omega$ having~$\bigo{p\cdot \omega}$ nodes can be computed in polynomial time.
\end{lemma}

It has long been known that calculating treewidth is {\nph} even for bipartite graphs~\cite{DBLP:journals/actaC/Bodlaender93}.
However, determining whether the treewidth of a graph is at most~$\omega$ is {\fpt} with respect to~$\omega$ and, moreover, powerful heuristic
and approximation algorithms for calculating treewidth have been perpetually reported~\cite{DBLP:conf/birthday/Bodlaender12,DBLP:journals/siamcomp/BodlaenderDDFLP16,DBLP:phd/basesearch/Zanden19}.
Besides, treewidth of chordal  bipartite graphs can be computed in polynomial time\footnote{A chordal bipartite graph is a bipartite graph without induced cycles of length at least six.}, and the gap between chordal bipartite  graphs and general bipartite graphs is quite small~\cite{DBLP:journals/jal/KloksK95}.
It is also well-known that every graph of~$p$ vertices admits an optimal tree decomposition of~$\bigo{p}$ nodes.
Then, by Lemma~\ref{lem-nice-tree-decomposition}, when we study {\fpt}-algorithms with respect to treewidth or a combined parameter involving treewidth, it does not lose any generality to assume that a nice tree decomposition is given.

\begin{theorem}
\label{thm-cav-fpt-treewidth}
{\wind{CCAV}} can be solved in time~$\bigos{4^{\omega}}$ if a nice tree decomposition of width~$\omega$ and of~$\emph{\textsf{poly}}(p)$ nodes of the incidence graph of the input election is given, where~$p$ is the number of vertices of the incidence graph.
\end{theorem}

\begin{proof}
Let $(E, k,\R)$ be an instance of {\wind{CCAV}} where $E=(C, V)$ and $k\leq \abs{C}$. In addition, let~$(T, \mathcal{B})$ be a nice tree decomposition of~$G_E$ of width~$\omega$ that has~$\textsf{poly}(p)$ nodes, where $p=\abs{C}+\abs{V}$.
We design a dynamic programming algorithm running in time~$\bigos{4^{\omega}}$ as follows.

For each bag $B_x\in \mathcal{B}$ associated with a node~$x$ in the tree~$T$, let~$C(B_x)$ and~$V(B_x)$ be the set of candidate-vertices and the set of vote-vertices contained in~$B_x$, respectively. Moreover, let~$C(T_x)$ be the set
of candidates in bags associated with nodes in the subtree rooted at~$x$. Obviously, $C(B_x)\subseteq C(T_x)$.
We maintain for each node~$x$ in~$T$ a $3$-dimensional table $D_x(C', V', k')$, where~$C'$ is a subset of~$C(B_x)$ such that $\abs{C'}\leq k$,~$V'$ is a submultiset of~$V(B_x)$, and~$k'$ is an integer such that $\abs{C'}\leq k'\leq \min\{k, \abs{C(T_x)}\}$.
We say that a $k'$-committee~$\ww$ is a valid committee for the entry $D_x(C', V', k')$ if all the following conditions hold simultaneously:
\renewcommand\labelenumi{(\theenumi)}
\begin{enumerate}
\item $\ww\subseteq C(T_{x})$;
\item $C'=C(B_x)\cap \ww$; and
\item a vote in~$V(B_x)$ has a representative in~$\ww$ if and only if this vote is in~$V'$,
i.e., for every $v\in V(B_x)$ it holds that $v\cap \ww\neq \emptyset$ if and only if $v\in V'$.
\end{enumerate}
The value of the entry $D_x(C', V', k')$ is the CCAV score of an optimal valid~$k'$-committee for the entry.
If~$D_x(C', V', k')$ admits no valid $k'$-committees, we define $D_x(C', V', k')=-\infty$.
Clearly, each table associated with a bag has at most $2^{\omega+1}\cdot (k+1)$ entries.

The algorithm updates the tables from those associated with the leaves up to the one associated with the root of~$T$. The values of entries associated with leaves are all~$0$
(recall that each leaf bag is empty). We update the entry $D_x(C', V', k')$ associated with a nonleaf node~$x$ in~$T$ as follows.

First, if some candidate from~$C'$ is approved by some vote from $V(B_x)\setminus V'$, we set $D_x(C', V', k')=-\infty$. Otherwise, we consider the following cases with respect to the types of~$x$.

\begin{itemize}
\item $x$ is a join node

 Let~$y$ and~$z$ be the two children of~$x$. We have $B_x=B_y=B_z$.
In this case, we let
    \[D_x(C', V', k')=
    \max_{\substack{k'_1+k'_2=k'-\abs{C'}\\ 0\leq k'_1\leq \abs{C(T_y)\setminus C'}\\  0\leq k'_2\leq \abs{C(T_z)\setminus C'}\\ V_1', V_2'\subseteq V', V_1'\cup V_2'=V'}}
    \left\{D_y\left(C', V_1', k'_1+\left|C'\right|\right)+D_z\left(C', V_2', k'_2+\left|C'\right|\right)-\left|V_1'\cap V_2'\right|\right\}.\]
It is obvious that there are~$\bigos{2^{\omega}}$ different combinations of~$k_1'$,~$k_2'$,~$V_1'$, and~$V_2'$ to consider in the max function. As a result, $D_x(C', V', k')$ can be computed in~$\bigos{2^{\omega}}$ time.

\item $x$ is an introduce node

 Let~$y$ be the child of~$x$, and let $\{h\}=B_x\setminus B_y$. We further distinguish between two subcases.

 \begin{itemize}
     \item $h$ is a vote

     If $h\not\in V'$, then due to the above discussion~$h$ does not approve any candidates from~$C'$. We set $D_x(C', V', k')=D_y(C', V', k')$. If $h\in V'$, then we set $D_x(C', V', k')=D_y(C', V'\setminus \{h\}, k')+1$ if~$h$ approves at least one candidate from~$C'$; and set $D_x(C', V', k')=-\infty$ otherwise.

     \item $h$ is a candidate

     If $h\not\in C'$, we set $D_x(C', V', k')=-\infty$ when $k'=\abs{C(T_x)}$, and set $D_x(C', V', k')=D_y(C', V', k')$ when $k'<\abs{C(T_x)}$.

     If $h\in C'$, let~$V_x(h)$ be the submultiset of votes in~$V'$ approving~$h$.
Then, we set
    \[D_x\left(C', V', k'\right)=\max_{U\subseteq V_x(h)}\left\{D_y\left(C'\setminus \left\{h\right\}, V'\setminus U, k'-1\right)+\abs{U}\right\}.\]
    Note that in this case if $V_x(h)=\emptyset$, we have that $D_x(C', V', k')=D_y(C'\setminus \{h\}, V', k'-1)$.
    As $\abs{V_x(h)}\leq \omega$, $D_x(C', V', k')$ can be calculated in $\bigos{2^{\omega}}$-time.
 \end{itemize}

\item $x$ is a forget node

 Let~$y$ be the child of~$x$ and let $\{h\}= B_y\setminus B_x$.
We consider the following two subcases.

\begin{itemize}
    \item $h$ is a vote

    We set $D_x\left(C', V', k'\right)=\max\left\{D_y\left(C', V', k'\right), D_y\left(C', V'\cup\left\{h\right\}, k'\right)\right\}$.

    \item $h$ is a candidate

Clearly, $h\not\in C'$. If $\abs{C'}=k'$ or~$h$ is approved by some votes in $V(B_x)\setminus V'$, we set $D_x(C', V', k')=D_y(C', V', k')$;
otherwise, let
\[D_x\left(C', V', k'\right)=\max\left\{D_y\left(C', V', k'\right), D_y\left(C'\cup \left\{h\right\}, V', k'\right)\right\}.\]
\end{itemize}
\end{itemize}

By the definition of the table, $D_r(\emptyset, \emptyset, k)$ is the CCAV score of an optimal valid~$k$-committee  for the entry, where~$r$ is the root of~$T$ (recall that the bag associated with the root is empty). By Conditions~(1)--(3) given above, every $k$-committee of~$C$ is valid for $D_r(\emptyset, \emptyset, k)$. Therefore, after all tables are computed, we conclude that the given instance of {\wind{CCAV}} is a {\yesins} if and only if $D_r(\emptyset, \emptyset, k)\geq \R$.

It remains to analyze the running time of the algorithm.
Each node in the nice tree decomposition is associated with a table of~$\bigos{2^{\omega}}$ entries.
Due to the above procedure, calculating an entry corresponding to a forget node takes polynomial time,
and calculating an entry corresponding to a join or an introduce node takes~$\bigos{2^{\omega}}$ time.
As the nice tree decomposition has polynomially many nodes, the running time of the algorithm is~$\bigos{4^{\omega}}$.
\end{proof}

Similar algorithms for {\wind{PAV}} and {\wind{MAV}} can be derived but with the running
time being bounded by~$\bigos{{(k+1)^{\omega}}}$. The reason is that in these cases we need to maintain more information in order to solve the problem.
For example, for {\wind{PAV}}, we need to maintain for each vote in~$V'$ the number of its approved candidates that are supposed to be in a desired committee, which requests a table of size at least $\bigos{(k+1)^{\omega}}$.

\begin{theorem}
\label{pav-mav-fpt-k-tree-width}
{\wind{PAV}} and {\wind{MAV}} are {\memph\fpt} with respect to the combined parameter $k+\omega$, if a nice tree decomposition of width~$\omega$ having~$\emph{\textsf{poly}}(p)$ nodes of the incidence graph of the input election is given, where~$p$ is the number of vertices of the incidence graph.
\end{theorem}

\begin{proof}
To prove the theorem, we derive {\fpt}-algorithms for {\wind{PAV}} and {\wind{MAV}} with respect to $k+\omega$. Let $I=(E,k,\R)$ be an instance of {\wind{$\tau$}} where $\tau\in \{\text{PAV}, \text{MAV}\}$, $E=(C, V)$, and $k\leq \abs{C}$. In addition, let~$p=\abs{C}+\abs{V}$, and let~$(T, \mathcal{B})$ be a nice tree decomposition of~$G_E$ of width~$\omega$ that has~$\textsf{poly}(p)$ nodes. The algorithms have the same skeleton as the one in the proof of Theorem~\ref{thm-cav-fpt-treewidth} for {\wind{CCAV}}. More concretely, we maintain a table for each node of~$T$, compute the tables in a bottom-up manner, and determine if the given instance~$I$ is a {\yesins} according to the table associated with the root.
To delineate the algorithms, we need the following notations. For an integer~$i$, we define $f(i)=\sum_{j=1}^i\frac{1}{j}$ if $i\geq 1$ and define $f(i)=0$ otherwise. For a function $\varrho: A\rightarrow B$ and an element $a\in A$, $\varrho_{-a}$ denotes~$\varrho$ restricted to the domain $A\setminus \{a\}$, i.e., $\varrho_{-a}: A\setminus \{a\}\rightarrow B$ is a function such that for all $a'\in A\setminus \{a\}$ it holds that $\varrho_{-a}(a')=\varrho(a')$. Additionally, let~$C(B_x)$,~$V(B_x)$, and~$C(T_x)$ be defined as in the proof of Theorem~\ref{thm-cav-fpt-treewidth}. Furthermore, in parallel with~$C(T_x)$, we let~$V(T_x)$ denote the multiset of all votes contained in bags associated with nodes in the subtree rooted at~$x$.
\medskip

{\noindent{\textbf{PAV.}}} For each node~$x$ in~$T$, we maintain a table $D_x(C', k', \mu)$ where $C'\subseteq C(B_x)$, $\abs{C'}\leq k'\leq \min\{k, \abs{C(T_x)}\}$, and~$\mu: V(B_x)\rightarrow [k]\cup \{0\}$ is a function. We note that if~$V(B_x)=\emptyset$,~$\mu$ is an empty function.
We say that a $k'$-committee $\ww\subseteq C(T_x)$ is valid for $D_x(C', k', \mu)$ if Conditions~(1) and~(2) listed in the proof of
Theorem~\ref{thm-cav-fpt-treewidth}, and the following condition are satisfied simultaneously:
\begin{enumerate}
    \item[(4)] $\abs{v\cap \ww}=\mu(v)$ for all $v\in V(B_x)$.
\end{enumerate}
We define $D_x(C', k', \mu)$ as the maximum possible PAV score of $k'$-committees valid for $D_x(C', k', \mu)$ with respect to~$V(T_x)$. More precisely,
\[D_x(C', k', \mu)\coloneqq \max_{\substack{\ww\subseteq C(T_x), \abs{\ww}=k'\\ \ww~\text{is valid for}~D_x(C', k', \mu)}}\textsf{PAV}(V(T_x), \ww),\]
if $D_x(C', k', \mu)$ admits at least one valid $k'$-committee, and $D_x(C', k', \mu)=-\infty$ otherwise.

The tables for the leaves can be computed trivially according to the definition of the tables. We show how to update an entry $D_x(C', k', \mu)$ by distinguishing the types of the node~$x$.

\begin{itemize}\itemsep=5pt
\item $x$ is a join node

 Let~$y$ and~$z$ be the two children of~$x$. We let
\begin{equation*}
\begin{split}
D_x(C', k', \mu)=  \max_{k'_1, k'_2, \mu_1, \mu_2} & \{D_y\left(C', k'_1+\left|C'\right|, {\mu_1'}\right)-g(\mu_1')\\
& +D_z\left(C', k'_2+\abs{C'}, \mu_2'\right)-g(\mu_2')\\
& +g(\mu)\},\\
\end{split}
\end{equation*}
where
\begin{enumerate}
    \item[(a)] $k'_1$ and~$k'_2$ run over all integers so that $k_1'+k_2'=k'-\abs{C'}$, $0\leq k'_1\leq \abs{C(T_y) \setminus C'}$, and $0\leq  k'_2\leq \abs{C(T_z) \setminus C'}$;
    \item[(b)] $\mu_1$ and $\mu_2$ run over all functions from~$V(B_x)$ to nonnegative integers so that for all $v\in V(B_x)$ it holds that $\mu_1(v)+\mu_2(v)=\mu(v)-\abs{v\cap C'}$;
    \item[(c)] for each $i\in [2]$,~$\mu_i'$ is a function from~${V(B_x)}$ to $[k]\cup\{0\}$ so that $\mu_i'(v)=\mu_i(v)+\abs{v\cap C'}$ for all $v\in V(B_x)$; and
    \item[(d)] for each $\rho\in \{\mu, \mu_1', \mu_2'\}$,  $g(\rho)=\sum_{v\in V(B_x)} f(\rho(v))$.
\end{enumerate}
By Condition~(4), the number of different combinations of~$\mu_1$ and~$\mu_2$ in~(b) is bounded from above by~$\bigos{(k+1)^{\omega+1}}\cdot \bigos{(k+1)^{\omega+1}}=\bigos{k^{2\omega}}$. As a result, an entry in this case can be computed in time~$\bigos{k^{2\omega}}$.
(Note that by the above recursion, when $V(B_x)=\emptyset$ we have that $D_x(C', k', \mu)=  \max_{k'_1, k'_2} \{D_y(C', k'_1+\abs{C'}, \mu)+D_z(C', k'_2+\abs{C'}, \mu)\}$.)

\item $x$ is an introduce node

 Let~$y$ be the child of~$x$, and let $\{h\}=B_x\setminus B_y$. We further distinguish between the following two subcases.

\begin{itemize}
    \item $h$ is a vote

    In this case, we set
    \[
    D_x(C', k', \mu)=
    \begin{cases}
        -\infty, & \text{if}~\mu(h)\neq \abs{h\cap C'}\\
        D_y(C', k', \mu_{-h})+f(\mu(h)), & \text{otherwise}\\
    \end{cases}
    \]

    \item $h$ is a candidate

    If $h\not\in C'$, we set $D_x(C', k', \mu)=-\infty$ when $k'=\abs{C(T_x)}$, and set $D_x(C', k', \mu)=D_y(C', k', \mu)$ when $k'<\abs{C(T_x)}$.

    If $h\in C'$, we set $D_x(C', k', \mu)=-\infty$ if there exists $v\in V(B_x)$ such that $h\in v$ and $\mu(v)=0$, and set $D_x(C', k', \mu)=D_y(C'\setminus \{h\}, k'-1, \mu')+\sum_{v\in V(B_x), h\in v}\frac{1}{\mu(v)}$ otherwise,
    where $\mu': V(B_y)\rightarrow [k]\cup \{0\}$ is the function so that for all $v\in V(B_y)$ it holds that $\mu'(v)=\mu(v)-\abs{v\cap \{h\}}$. (If~$V(B_x)=\emptyset$,~$\mu'$ is an empty function.)
\end{itemize}

\item $x$ is a forget node

 Let~$y$ be the child of~$x$ and let $\{h\}= B_y\setminus B_x$.
Similar to the above case we distinguish between the following two subcases.
\begin{itemize}
    \item $h$ is a vote

    We set $D_x(C', k', \mu)=\max_{\substack{\mu': V(B_y)\rightarrow [k]\cup \{0\}\\ \mu'_{-h}=\mu}}D_y(C', k', \mu')$. Obviously, we have at most~${(k+1)^{\omega+1}}$ different functions~$\mu'$ to check, and hence such an entry $D_x(C', k', \mu)$ can be computed in $\bigos{{k}^{\omega}}$ time.

    \item $h$ is a candidate

    In this case, we set
    \[
    D_x(C', k', \mu)=
    \begin{cases}
        D_y(C', k', \mu), & \text{if}~\abs{C'}=k'\\
        \max\{D_y(C', k', \mu), D_y(C'\cup\{h\}, k', \mu)\}, & \text{otherwise}\\
    \end{cases}
    \]
\end{itemize}
\end{itemize}

After the table for the root~$r$ is computed, we conclude that the given instance~$I$ is a {\yesins} if and only if $D_r(\emptyset, k, \mu)\geq \R$, where~$\mu$ is an empty function (recall that the bag associated with the root is empty).
This is because that by the definition of the table, $D_r(\emptyset, k, \mu)$ is the maximum possible PAV score of a  valid $k$-committee for $D_x(\emptyset, k, \mu)$ with respect to~$V(T_r)=V$.
\medskip

{\noindent{\textbf{MAV.}}} For each node~$x$ in~$T$, we maintain a table $D_x(C', k', \mu)$ whose components are defined as with the ones for PAV. The validity of a $k'$-committee for an entry is also defined the same.
However, in this case each entry takes only binary values~$1$ and~$0$. The entry $D_x(C', k', \mu)$ is~$1$ if and only if there is at least one $k'$-committee~$\ww$ which is valid for the entry (i.e., a~$k'$-committee satisfying Conditions~(1)--(2) in the proof of Theorem~\ref{thm-cav-fpt-treewidth}, and Condition~(4) given above) and, moreover, for each $v\in V(T_x)\setminus V(B_x)$, it holds that $\abs{\ww\cap v}\geq \frac{k+\abs{v}-\R}{2}$. Observe that, by the definition of nice tree decomposition, none of the votes in $V(T_x)\setminus V(B_x)$ approves any candidates from $C\setminus C(T_x)$. Therefore, for every $k$-committee $\ww'\subseteq C$ such that $\ww'\cap C(T_x)=\ww$ and every $v\in V(T_x)\setminus V(B_x)$, it holds that $v\cap \ww=v\cap \ww'$, implying that the Hamming distance between~$v$ and~$\ww'$ is at most~$\R$ if and only if $\abs{v\cap \ww}\geq \frac{k+\abs{v}-\R}{2}$. The requirement $\abs{\ww\cap v}\geq \frac{k+\abs{v}-\R}{2}$ ensures that as long as the given instance admits a {\yes}-witness~$\ww'\subseteq C$ such that $\ww'\cap C(T_x)=\ww$, the existence of this {\yes}-witness is safely preserved in the table.

The tables for the leaves can be computed trivially according to their definitions. We show how to update an entry $D_x(C', k', \mu)$ by distinguishing the types of the node~$x$.

\begin{itemize}
\item $x$ is a join node

Let~$y$ and~$z$ be the two children of~$x$. In this case, we set $D_x(C', k', \mu)=1$ if and only if there are two functions $\mu_1, \mu_2: V(B_x)\rightarrow [k]\cup \{0\}$ and two nonnegative integers~$k_1'$ and~$k_2'$ such that
\begin{enumerate}
    \item[(a)] $k'_1+k'_2=k'-\abs{C'}$;
    \item[(b)] for all $v\in V(B_x)$ it holds that $\mu_1(v)+\mu_2(v)=\mu(v)-\abs{v\cap C'}$ and, moreover,
    \item[(c)] $D_y\left(C', k'_1+\abs{C'}, {\mu_1'}\right)=D_z\left(C', k'_2+\abs{C'}, \mu_2'\right)=1$, where for each $i\in [2]$,~$\mu_i'$ is the function from~${V(B_x)}$ to $[k]\cup\{0\}$ so that $\mu_i'(v)=\mu_i(v)+\abs{v\cap C'}$ for all $v\in V(B_x)$.
\end{enumerate}
It is obvious that there are~$\bigos{k^{\bigo{\omega}}}$ different combinations of~$k_1'$,~$k_2'$,~$\mu_1$, and~$\mu_2$ to enumerate. As a result, it takes $\bigos{k^{\bigo{\omega}}}$ time to compute $D_x(C', k', \mu)$.
(Note that when $V(B_x)=\emptyset$, the above recursion indicates that $D_x(C', k', \mu)=1$ if and only if there are~$k_1'$ and~$k_2'$ as above so that $D_y(C', k_1'+\abs{C'},\mu)=D_z(C', k_2'+\abs{C'},\mu)=1$.)

\item $x$ is an introduce node

 Let~$y$ be the child of~$x$, and let $\{h\}=B_x\setminus B_y$. We further distinguish between the following two subcases.

\begin{itemize}
    \item $h$ is a vote

    We set $D_x(C', k', \mu)=1$ if and only if $\mu(h)=\abs{C'\cap h}$ and $D_y(C', k', \mu_{-h})=1$.

    \item $h$ is a candidate

    If $h\not\in C'$, we set $D_x(C', k', \mu)=0$ when $k'=\abs{C(T_x)}$, and set $D_x(C', k', \mu)=D_y(C', k', \mu)$ when $k'<\abs{C(T_x)}$.

    If $h\in C'$, we set $D_x(C', k', \mu)=1$ if and only if
    $D_y(C'\setminus \{h\}, k'-1, \mu')=1$ where $\mu': V(B_y)\rightarrow [k]\cup \{0\}$ is a function so that for all $v\in V(B_y)$ it holds that $\mu'(v)=\mu(v)-\abs{v\cap \{h\}}$. (When $V(B_x)=\emptyset$, $D_x(C', k', \mu)=D_y(C'\setminus \{h\}, k'-1, \mu)$.)
\end{itemize}

\item $x$ is a forget node

 Let~$y$ be the child of~$x$ and let $\{h\}= B_y\setminus B_x$.
\begin{itemize}
    \item $h$ is a vote

    We set $D_x(C', k', \mu)=1$ if and only if
    \begin{enumerate}
        \item[(1)] $\mu(h)\geq \frac{k+\abs{h}-\R}{2}$, and
        \item[(2)] there is a function $\mu': V(B_y)\rightarrow [k]\cup \{0\}$ so that $\mu=\mu'_{-h}$ and $D_y(C', k', \mu')=1$.
    \end{enumerate}
    The first condition is to ensure that there is a $k$-committee which contains a valid $k'$-committee~$\ww$ for the entry and the Hamming distance between~$\ww$ and~$h$ is at most~$\R$.

    \item $h$ is a candidate

    If $\abs{C'}=k'$, we set $D_x(C', k', \mu)=D_y(C', k', \mu)$.
    Otherwise, we set $D_x(C', k', \mu)=1$ if and only if $D_y(C', k', \mu)+D_y(C'\cup\{h\}, k', \mu)\geq 1$.
\end{itemize}
\end{itemize}
After the table for the root~$r$ is computed, we conclude that the given instance~$I$ is a {\yesins} if and only if $D_r(\emptyset, k, \mu)=1$, where~$\mu$ is an empty function. The reason for this is that by the definition of the table, it holds that $D_r(\emptyset, k, \mu)=1$ if and only if there is at least one $k$-committee $\ww\subseteq C(T_r)=C$ which is valid for $D_r(\emptyset, k, \mu)$ and, moreover, for each $v\in V(T_r)\setminus V(B_r)=V$, it holds that $\abs{\ww\cap v}\geq \frac{k+\abs{v}-\R}{2}$.

The running times of the algorithms for {\wind{PAV}} and {\wind{MAV}} are dominated by the total size of all tables which is bounded from above by $\bigos{k^{\bigo{\omega}}}$.
\end{proof}

\subsection{Maximum Matching}
It is well-known that for every bipartite graph~$G$, the size of a maximum matching of~$G$ equals the size of a minimum vertex cover of~$G$~\cite{LPmatchingtheory1986}. As the treewidth of
a graph is bounded from above by  the size of a minimum vertex cover of the graph~\cite{DBLP:journals/algorithmica/FominLMT18}, it follows from Theorem~\ref{thm-cav-fpt-treewidth} that {\wind{CCAV}} is {\fpt} with respect to~$\alpha(G_E)$, where~$E$ is a given election. We show that {\wind{MAV}} and {\wind{PAV}} are also {\fpt} with respect to this parameter.
Observe that the size of a maximum matching of~$G_E$ can be at most~$\min\{m, n\}$, where~$m$ is the number of candidates and~$n$ is the number of votes in the election~$E$. Consequently, every {\fpt}-algorithm with respect to~$\alpha(G_E)$ also runs in {\fpt}-time in~$n$ and~$m$.

\begin{theorem}
\label{thm-mav-pav-matching}
{\wind{MAV}} and {\wind{PAV}} are {\memph\fpt} with respect to~$\alpha(G_E)$, where~$E$ is the election in the input.
\end{theorem}

\begin{proof}
Let $I=(E, k, \R)$ be an instance of {\wind{$\tau$}}, where $E=(C, V)$ and $\tau\in \{\text{MAV}, \text{PAV}\}$. Let~$G_E$ be the incidence graph of~$E$, and let~$M$ be a maximum matching of~$G_E$.
Hence, $\alpha(G_E)=\abs{M}$. For simplicity, we write~$\alpha$ for~$\alpha(G_E)$. Let~$C(M)$ and~$V(M)$ be the set of candidates and the multiset of votes saturated by~$M$, respectively. Obviously, $\abs{C(M)}=\abs{V(M)}=\alpha$.
We derive algorithms for MAV and PAV as follows.
\medskip

\noindent{\bf{MAV.}}
We split~$I$ into at most~$2^{\alpha}$ subinstances, each of which takes~$I$ and a subset $C'\subseteq C(M)$ of size at most~$k$ as input, and
asks whether there is a~$k$-committee~$\ww\subseteq C$ of MAV score at most~$\R$ such that $C'= C(M)\cap \ww$. It is easy to see that~$I$ is a {\yesins} if and only if at least one of the subinstances is a {\yesins}.
We show how to solve each subinstance in polynomial time.

Let $I'=(I, C')$ be a subinstance. Let~$v$ be a vote in $V\setminus V(M)$. Assume that~$\ww$ is a~$k$-committee such that $\ww\cap C(M)=C'$.
As~$M$ is a maximum matching of~$G_E$, no vote from $V\setminus V(M)$ approves any candidate from $C\setminus C(M)$. Consequently, the Hamming distance between~$v$ and~$\ww$ is $\abs{v}+k-2\abs{v\cap C'}$.
Hence, if there is a vote in $V\setminus V(M)$ such that $\abs{v}+k-2\abs{v\cap C'}>\R$, the subinstance is a {\noins}. So, let us assume that this is not the case.
Now the task is to identify a subset~$H$ of $k-\abs{C'}$ candidates from $C\setminus C(M)$ such that the Hamming distance between every vote from~$V(M)$ and the committee $H\cup C'$ is at most~$\R$, or equivalently, for every $v\in V(M)$,~$H$ contains at least $\frac{\abs{v}+k-\R}{2}-\abs{v\cap C'}$ of~$v$'s approved candidates in $C\setminus C(M)$. For each $v\in V(M)$, let $f(v)=\frac{\abs{v}+k-\R}{2}-\abs{v\cap C'}$.
We reduce the subinstance into an ILP with a bounded number of variables. Specifically, for each $U\subseteq V(M)$, let~$C_U$ be the set of candidates in $C\setminus C(M)$ approved by all votes in~$U$ but not approved by any vote in $V(M)\setminus U$, and let $m_U=\abs{C_U}$. For each $U\subseteq V(M)$, we create one nonnegative integer variable~$x_{U}$ which indicates the number of candidates from~$C_U$ that are supposed to be in a desired committee.
The constraints are as follows.
\begin{itemize}
\item First, for each variable~$x_U$ we have that $0\leq x_U\leq m_U$.

\item Second, as we seek $k-\abs{C'}$ candidates from $C\setminus C(M)$, it holds that \[\sum_{U\subseteq V(M)}x_U=k-\abs{C'}.\]

\item Finally, for every $v\in V(M)$, it holds that
\[\sum_{U\subseteq V(M), v\in U}x_U\geq f(v).\]
This inequality ensures that the Hamming distance between every vote $v\in V(M)$ and the desired committee is at most~$\R$.
\end{itemize}
It is known that ILP is {\fpt} with respect to the number of variables~\cite{Lenstra83}. As the above ILP has at most~$2^{\alpha}$ variables, the subinstance can be solved in {\fpt}-time in~$\alpha$.
As there are at most $2^{\alpha}$ subinstances, the whole algorithm runs in {\fpt}-time in~$\alpha$.
\medskip

\noindent{\bf{PAV}.} The algorithm is analogous to the above one for MAV.
First, we split the given instance~$I$ into subinstances by enumerating all possible intersections~$C'$ of~$C(M)$ and a desired $k$-committee.
Then, every nonempty vote $v\in V\setminus V(M)$ such that $v\cap C'\neq \emptyset$ provides a PAV score
$\sum_{i=1}^{\abs{v\cap C'}}\frac{1}{i}$ to every $k$-committee whose intersection with~$C(M)$ is exactly~$C'$.
Let
\[\R'=\sum_{v\in V\setminus V(M), v\cap C'\neq \emptyset} \sum_{i=1}^{\abs{v\cap C'}}\frac{1}{i}.\] The question now is equivalent to solving an instance $((C, V(M)), C', k, \R-\R')$ of {\awind{PAV}}, which can be done in {\fpt}-time in $\abs{V(M)}=\abs{M}=\alpha$ (Theorem~\ref{thm-pav-n-fpt}).
\end{proof}

\section{Conclusion}
\label{sec-conclusion}
We have investigated the parameterized complexity of {\wind{$\tau$}} for $\tau$ being the three prevalent approval-based $k$-committee selection rules MAV,  CCAV,  and PAV, aiming at providing plentiful fixed-parameter tractability results with respect to meaningful parameters. We studied many natural single parameters, their combinations, and two structural parameters of incidence graphs of elections, and obtained an almost complete landscape of the parameterized complexity of {\wind{$\tau$}} for $\tau\in \{\text{MAV}, \text{CCAV}, \text{PAV}\}$ with respect to these parameters.
For a summary of our concrete results, we refer to Table~\ref{tab-results}.
It should be noted that many of our tractability results were obtained by reducing {\wind{$\tau$}} to well-studied graph/set problems. As advocated by several researchers~\cite{DBLP:conf/aaai/ElkindL14,Fernau14FLMPS}, a remarkable advantage of using
the reduction scheme is that we can automatically  update our results with the state-of-the-art of these graph/set problems. In addition, though many reductions are trivial and direct, we deem that pointing out the connections between {\wind{$\tau$}} and the graph/set problems benefits researchers from different communities.

Our exploration leaves several intriguing problems for future research. We select two open questions that we believe to be the most challenging.

\begin{openquestion}
Is {\wind{PAV}} {\memph\fpt} with respect to~$\R$?
\end{openquestion}

\begin{openquestion}
Is {\wind{PAV}} {\memph\fpt} with respect to the treewidth of the incidence graph of the given election?
\end{openquestion}

Besides these open questions, improving the {\fpt}-algorithms presented in the paper or investigating the kernelizations of {\fpt} problems are also promising avenues for future research.
We remark that Agrawal~et~al.~\cite{DBLP:conf/cocoon/AgrawalCJKS018} studied kernelizations of the {\sc{Partial Hitting Set}} problem. Their results imply that
 {\wind{CCAV}} admits a polynomial kernel with respect to $\R+\triangle_{\text{V}}$. It is interesting to see if the same holds for {\wind{PAV}}. Regarding MAV, Misra, Nabeel, and Singh~\cite{DBLP:conf/atal/MisraNS15} have shown that {\wind{MAV}} is unlikely to admit any polynomial kernels with respect to $\R+m$ and $n+k$, assuming standard complexity hypothesis. This means that {\wind{MAV}} is unlikely to admit any polynomial kernels with respect to smaller parameters (e.g.,~$\R$,~$n$, $\R+\triangle_{\text{V}}$, etc.).

Finally, we would like to mention that recently a few researchers have examined a number of structural parameters based on real electoral data or data generated by promising models, largely motivated by previous theoretical works on the parameterized complexity of voting problems with respect to these parameters (see, e.g.,~\cite{DBLP:journals/corr/abs-2204-03589,DBLP:conf/ijcai/SuiFB13}). Although graphs of small treewidth arisen from various research areas have been continually reported in the literature, to the best of our knowledge, analogous works on the treewidth of incidence graphs of elections have not been conducted heretofore. We hope that our algorithms presented in Section~\ref{sec-structural-parameter} would inspire such an investigation.

\section*{Appendix}
\label{sec-appendix}
This appendix is devoted to the proof of Lemma~\ref{lem-trivial-cases-pav}. In the following, when we remove a candidate from an election, we  remove it from the candidate set and from all votes approving the candidate.
\medskip

\noindent{\bf{Lemma~\ref{lem-trivial-cases-pav}.}}
{\it{Let $\mathcal{H}$ be the set of all paths, cycles, hairsticks, and DH-hairsticks. Then, given an $\mathcal{H}$-election, we can compute a PAV optimal $k$-committee of the election in polynomial time.}}
\medskip

\begin{proof}
Let $E=(C, V)$ be an election whose multihypergraph representation~$H(E)$ is either a path, a cycle, a hairstick, or a DH-hairstick. Let~$k$ be a nonnegative integer such that $k\leq \abs{C}$.
Let~$n=\abs{V}$. Note that~$n$ is also the number of vertices in~$H(E)$. We consider all the possible cases of~$H(E)$ as follows.

\begin{description}
\item[Case~1: $H(E)$ is a path or a cycle.] \hfill

We first compute a maximum matching~$M$ of~$H(E)$. If $\abs{M}\geq k$, it is easy to see that the $k$-committee corresponding to~$M$ is optimal. Otherwise, we further distinguish between two subcases based on the parity of~$n$.
Let~$\ww$ be the committee corresponding to~$M$.

If~$n$ is even, we know that~$M$ is a perfect matching. Let~$M'$ be the set of edges of~$H(E)$ without~$M$. Obviously,~$M'$ is also a matching of~$H(E)$. Let~$C(M')$ be the set of candidates corresponding to~$M'$. As~$M$ is a perfect matching, adding each $c\in C(M')$ into any committee containing~$\ww$ and excluding~$c$ increases the PAV score of the committee by exactly one. In light of this fact, we add any $k-\abs{M}$ arbitrary candidates from~$C(M')$ into~$\ww$, and return~$\ww$.

If~$n$ is odd, there is exactly one vertex~$v$ in~$H(E)$ not saturated by~$M$. We put into~$\ww$ one arbitrary candidate approved by~$v$. Then, we arbitrarily add $k-1-\abs{M}$ candidates from the remaining candidates into~$\ww$, and return~$\ww$.

\item[Case~2: $H(E)$ is a hairstick.] \hfill

If $k=n$, we return the whole set of candidates. Otherwise (i.e., $k<n$), a key observation is that there exist optimal $k$-committees which do not contain the candidate corresponding to the loop in~$H(E)$. By this observation, we directly remove the loop-candidate from the election. Now, the election admits a path representation, and we use the algorithm described in Case~1 to compute an optimal $k$-committee.

\item[Case~3: $H(E)$ is a DH-hairstick.] \hfill

If $k=n+1$, we return the whole set of candidates. Otherwise, there exist optimal $k$-committees which contain at most one of the two candidates corresponding to the loops in~$H(E)$. In this case, we arbitrarily select one loop, and remove the corresponding candidate from the election. Then, we arrive at a hairstick representation of the election. The algorithm given in Case~2 is applied to solve the problem.
\end{description}

It is easy to see that the above algorithms in all cases run in polynomial time.
\end{proof}

\section*{Acknowledgments}
This work was supported in part by the National Natural Science Foundation of China under Grant 62172446, and in part by the Open Project of Xiangjiang Laboratory under Grants 22XJ02002 and 22XJ03005.

The authors would like to thank all anonymous reviewers who have provided helpful comments on the paper.

\end{document}